\documentclass[12pt,a4paper]{article}
\usepackage{epsfig}
\usepackage{amsfonts,amssymb}
\usepackage{amsmath}
\usepackage{bm}
\usepackage{bbm}

\newtheorem{theorem}{Theorem}[section]

\newtheorem{definition}[theorem]{Definition}
\newtheorem{example}[theorem]{Example}

\newtheorem{lemma}[theorem]{Lemma}

\newtheorem{proposition}[theorem]{Proposition}

\newenvironment{proof}[1][Proof]{\noindent\textbf{#1.} }{\ \rule{0.5em}{0.5em}}
\numberwithin{equation}{section}

\begin{document}

\title{Generalized Green Functions and current correlations in the TASEP}
\author{A.M. Povolotsky$^{1,\dag}$, V.B. Priezzhev$^{1,\ddag}$, G.M. Sch{\"u}tz$^{2,\sharp}$}
\maketitle

\abstract{We study correlation functions of  the totally asym\-metric simple exclusion process (TASEP) in discrete time
with back\-ward sequen\-tial update. We prove a determinantal formula for the generalized Green function which describes
transitions between positions of particles at different individual time moments. In particular, the generalized
Green function defines a probability measure at staircase lines on the space-time plane.
The marginals of this measure are the TASEP correlation functions in the space-time region not covered by the standard
Green function approach. As an example, we calculate the current correlation function that is the joint probability
distribution of times taken by selected particles to travel given distance.
An asymptotic analysis  shows that current fluctuations converge to the ${Airy}_2$ process.}
\\  \\

{\footnotesize \noindent $^{~1}$Bogoliubov Laboratory of Theoreticl Physics,
Joint Institute for Nuclear Research, Dubna 141980, Moscow Region, Russia}
\\ \\
{\footnotesize \smallskip \noindent $^{~2}$Forschungszentrum J\"ulich GmbH,
Institut f\"ur Festk\"orperforschung, D--52425 J\"ulich, Germany, and \newline
Interdisziplin\"ares Zentrum f\"ur Komplexe Systeme, Universit\"at Bonn,
Ger\-ma\-ny}
\\ \\
\noindent e-mail addresses:\\ \begin{minipage}{10cm} $^\dag$\texttt{alexander.povolotsky@gmail.com},\\$^\ddag$\texttt{priezzvb@theor.jinr.ru},\\ $^\sharp$\texttt{g.schuetz@fz-juelich.de}\end{minipage}

\section{Introduction}

The totally asymmetric simple exclusion process (TASEP) is one of the most studied models of non-equilibrium
statistical mechanics \cite{Spohn,Derrida,Liggett,SchutzBook}. A rough guide in this field would distinguish
between continuous time and discrete time processes. The latter are subdivided into several groups in accordance
with updating rules: parallel, sequential, sublattice parallel, etc.\cite{update}. The detailed microscopic
information about the time evolution in the TASEP is provided by the transition probability or the Green function
of the master equation. The Green function describing the continuous time evolution of $N$ particles on the
infinite one-dimensional lattice   has been obtained in the form of  a determinant of a $N\times N$ matrix in
\cite{Sch1}. For periodic boundary conditions the Green function is given by the $N$-fold sum of similar
determinants \cite{priezPRL}. The determinantal expression of the Green function for the parallel update has been
found in \cite{PovPri1} by the Bethe ansatz techniques and in \cite{BFS} by induction
in time. The ring geometry for the parallel update was considered in \cite%
{PovPri2}. The model we deal with in this paper is the discrete time TASEP with sequential update. A formal
derivation of the determinant for the Green function for this process has been undertaken in
\cite{priezzhev_traject} by a geometrical treatment of the Bethe ansatz and in \cite{RS} directly using
the Bethe ansatz.

The Green function defines a probability  measure on the set of particle configurations.  The correlation
functions, i.e. distributions of positions of selected particles, are  marginals of this measure. Recently,
considerable progress has been achieved in the calculation of correlation functions. In the pioneering work by
 Johansson,  \cite{Johansson},  the distribution of the position  of a tagged particle  at a
given time was obtained for the TASEP with parallel update and  step initial conditions. The technique used involved
several mappings: from the TASEP to the last passage percolation problem, then to the problem of the longest
nondecreasing subsequence of a generalized permutation, and, finally, via Robinson-Schensted-Knuth's
correspondence,
 to the statistics of pairs of Young tableaux.
 In contrast, the same result can be obtained by direct summation of the
Green function over the whole set of final configurations constrained within
a suitably chosen domain. In this complementary way, Johansson's result was
generalized to the TASEP with backward sequential update in \cite{RS} and extended
to different initial conditions for the continuous time TASEP  in \cite{Nagao}.
The distributions
were presented in the form of Fredholm determinants with polynomial kernels and were proved to converge asymptotically
  to the universal
Tracy-Widom distributions of the largest eigenvalues in the Gaussian random matrix ensembles, unitary and orthogonal
for step and flat initial conditions respectively. The Tracy-Widom distributions were shown to be universal
scaling functions of the Kardar-Parisi-Zhang (KPZ) universality class \cite{KPZ}.

The direct summation of the Green function hardly helps for the multi-particle distributions. Instead, it was
observed in \cite{Sasamoto} that the determinantal formula for the Green function  dependent on $N$ particle
coordinates can be treated as the marginal of an auxiliary determinantal measure on much larger space of
$N(N-1)/2$ variables, which are identified with  coordinates of fictitious particles with free fermionic
interaction. The marginals of this measure are calculable using the construction of the determinantal point process,
from where the correlation functions of interest follow. Similarly to the one-particle distributions, they
 can be represented in the form of  Fredholm determinants and which asymptotically converge to
the universal scaling limits defining the  Airy$_1$, Airy$_2$, e.t.c. processes, which were claimed to be the
hallmarks of the KPZ class. The results for joint distributions of positions of several particles at a given time
have been obtained in   \cite{BFPS,BFP,BFS1} for different updates and initial conditions. The reduction to a
suitable determinantal process turned out to be efficient  also for the calculation of more general correlation
functions \cite{BorodinRains}. In particular,  the  chain of mappings similar to that from the Johansson's work
was used in \cite{ImamSasamoto} to find the auxiliary determinantal measure, a.k.a. the Schur process, which allowed a
calculation of the   distribution of positions of a tagged particle at different moments of time.

The most general correlation functions, which refer to the positions  of selected particles at selected, though
generally different, time moments, were obtained in \cite{BFS,BorodinFerrari}. The authors considered a
cluster of particles initially compact, hopping to the right. Its evolution was sliced into as many slices as
the number of selected time
moments, each slice connecting two subsequent time moments. Then the correlation function  was constructed as a
convolution of the  Green functions, where each Green function described the evolution within a single time slice and
fixed the position  of another selected particle.  Relying on the fact that in the TASEP the rightward motion of
particles is independent of the particles on the left of them, each Green function described only the evolution
of particles on the right of the particle position fixed at the previous step.  This structure stipulates a weak
ordering of the time moments within the correlation function: the more to the left is the initial position of the
particle, the earlier is the corresponding time moment. Here "weak" refers to the fact that the times
(particles) selected for different particles (times) are also allowed to  be  equal, which reproduces the known
cases of fixed time (tagged particle) correlations respectively. This ordering was named
space-like \cite{BFS,BorodinFerrari}. Such a multilevel
construction gave rise to even more complicated auxiliary determinantal point process than for the fixed time
case, which finally resulted in the correlation functions on space-like point configurations. Of course, the
limiting processes obtained from the asymptotical analysis of correlation functions at the space-like paths have
the expected KPZ-specific form.

The space-like ordering  used  in \cite{BFS,BorodinFerrari} is a technical condition originating from the
method used. Therefore, one expects that it would not be a constraint for the KPZ universality. Indeed, basing
on the observed slow decorrelation phenomena, it was proved in \cite{Ferrari,CorwinFerrariPeche} that the
limiting processes are  the same along any space-time path, except special directions given by the characteristic
lines of the hydrodynamic conservation law associated to the large scale particle dynamics.  However,  no
microscopic results on correlation functions beyond the distributions of particle positions on the space-like
paths were available up to date.

The aim of the present paper is to extend the set of known correlation functions by using a generalization of the
conventional Green function, which  describes transitions between space-time point configurations where different
particles are associated with different  time moments and  different space positions. Such a generalization was
proposed in \cite{Brankov}. It was shown by means of a trajectory treatment of the Bethe ansatz that, under
special conditions on space-time points of initial and final configurations, the generalized Green function has a
determinantal representation similar to that of the regular Green function. The idea now is to supply the
generalized Green function with a meaning of probability measure on some set of space-time point configurations
passed during  the TASEP evolution, and to calculate the TASEP correlation functions
as  the marginals of this  measure. We partially implement this program here. We define the space-time point sets
on staircase lines  named boundaries by requiring them to satisfy two conditions. First, the Green function
describing the transition
between point configurations within  these sets is of the determinatal form. Second, the TASEP
evolution  has a Markov property at the boundary needed to claim that the probability of a space-time
trajectory of the system factorizes into the product of probabilities of the two parts obtained by cutting the
particle trajectories along the boundary. Due to this fact, the required probability measure can be
defined as the probability of configurations of space-time points from where the particle trajectories go
after they have traversed the boundary. Remarkably, the time coordinates of points in the configurations available at the
boundaries are ordered in an opposite way than those in the space-like configurations, and, hence, we have a new
class of configurations to deal with. They, however, do not provide a full complement to the set of space-like
configurations from  \cite{BFS,BorodinFerrari}, because there is also a restriction on the space coordinates of the
points, which are required to be strictly decreasing with the number of a particle.

After discussing the general construction, giving the probability measure on the new class of configurations, we
still need to present a determinantal process to proceed with the correlation functions. The further analysis
is limited to a particular example of the generalized boundary, where $N$ particles exit from $N$ fixed subsequent
space points at arbitrary time.  For the step initial conditions, this suggests that all particles make the same
number of steps. Then, the correlation functions of interest are the joint distributions of times taken by
selected particles to travel a given distance.  For this case we obtain the corresponding determinantal measure
and complete the calculation of correlation functions. We should note that similar  auxiliary determinantal
measure for the same problem in the case of parallel update was  proposed in \cite{ImamSasamoto}, where the term
current correlations was used. Therefore, we keep on the same terminology here. The auxiliary determinantal
measure was obtained in \cite{ImamSasamoto} using the above mentioned chain of mappings. This method is
restricted to the step initial conditions only. At the same time, the generalized Green function give a more
general tool applicable to a wider class of initial conditions and finial configurations. Therefore, the results
of the present paper on current correlation functions should be viewed  as a test example of the method used. The more
general cases will be considered elsewhere.

The paper is organized as follows. In Section 2 we define the discrete time TASEP with sequential update and
present some basic properties. Then we proceed with the derivations of the Green functions. To be self-contained,
we start with a regular derivation of the Green function for sequential update, using an inductive solution of
the Chapman-Kolmogorov equation. This is done in Section 3. In the Section 4, we introduce the
generalized  Green function. In the first part of this section, Proposition \ref{FDisc}, we present a rigorous
proof of the determinantal formula first proposed in \cite{Brankov}. The formula is restricted to the case when
initial and final point configurations are of the special form, which we call admissible, defined by the ordering
of coordinates in time and space. In the second part of Section 4 we introduce the boundary sets, where the
admissible configuration live. We prove  that the generalized Green functions define the
probability measure at the boundaries and possesses the Markov property, implying that the convolution of two Green functions performed over
the boundary yields another Green function. These results are summarized in  Propositions \ref{GGF prob}
 In Section 5, we use the Sasamoto's formalism of auxiliary variables \cite{Sasamoto} to rewrite
the generalized Green function in the form of a signed determinantal point process and write its correlation
kernel in the Proposition \ref{det_measure}. The main results of the paper on the current correlations are given
in Section 6.  Theorem \ref{fredholm} gives  the joint distribution of particle currents in the form of
Fredholm determinant. The asymptotic form of this determinant and convergence to the $Airy_2$ process is the
content of  Theorem \ref{Airy_2}. The latter is given without a detailed  mathematical proof. Only the sketch
of the saddle-point analysis of the kernel is presented in  Lemma \ref{scaling limit}. A conclusion
and a review of perspectives are in Section 7. In particular an application of the generalized Green functions to
more general situations is briefly discussed.

\section{Discrete time TASEP}

Consider $N$ particles on the one-dimensional integer lattice. A configuration of the system
\begin{equation}
\bm{x}\in \mathbb{Z}^N_>
\end{equation}
takes its values in the set
\begin{equation}
\mathbb{Z}^N_>= \left\{(x_1,\dots,x_N) \in \mathbb{Z}^N : x_{1}>x_{2}>\dots
>x_{N}\right\},  \label{range}
\end{equation}%
which is an $N$-tuple of strictly increasing coordinates of particles $%
(x_1,\dots,x_N)$. The strictly increasing order reflects the exclusion condition that two particles cannot occupy
the same site.

The Markovian dynamics of this process for a single discrete time step can be defined as follows. Starting from
the rightmost particle at position $x_1 < \infty$, we let each particle jump to its right neighbouring site with
probability $p$, provided that the target site is empty. If the target site is occupied, the probability for a
particle to stay is $1$. Therefore the
TASEP is a random process which is given as a sequence of configurations $%
\bm{x}^{0},\bm{x}^{1},\dots ,\bm{x}^{t}$. We refer to such a sequence as a trajectory of the system up to time
$t$. Every trajectory is realized with probability
\begin{equation}
P(\bm{x}^{0},\dots ,\bm{x}^{t})=P_1(\bm{x}^{t}|\bm{x}^{t-1})\dots P_1(\bm{x}%
^{2}|\bm{x}^{1})P_1(\bm{x}^{1}|\bm{x}^{0})P_{0}(\bm{x}^{0}),
\end{equation}%
where $P_{0}(\bm{x})$ is the initial probability of configuration $\bm{x}$.
The one-step transition probability $P_1(\bm{x}|\bm{y})$ from configuration $%
\bm{y}$ to $\bm{x}$ takes the form
\begin{equation}
P_1(\bm{x}|\bm{y})=\prod\limits_{i=1}^{N}\theta \left( x_{i}-y_{i},x_{i-1}-y_{i}\right) ,  \label{weights}
\end{equation}%
where
\begin{eqnarray}
\theta \left( k,l\right) &=&\left( q+p \delta_{l,1}\right) \delta _{k,0}+p\delta _{k,1},  \label{weights_12}
\end{eqnarray}%
and we formally define $x_0=\infty$. The parameter $p$, the hopping probability, varies in the range
\begin{equation}
0 <p<1,  \label{p range}
\end{equation}%
and we also define
\begin{equation}
q=1-p.  \label{q}
\end{equation}%
In other words, we exclude the trivial deterministic limiting cases $p=0,1$ from our consideration. The above
transition probabilities correspond to the discrete time dynamics with backward sequential update, when at each
time step the position of the first particle is updated first, then of the second particle, etc. Since the
hopping is totally asymmetric, the first $N$ particles do not ``feel'' the presence of any particle to their
left. Hence, when considering the dynamics of the first $N$ particles, one may study without loss of generality
the system with exactly $N$ particles initially. Then for any fixed $t$ the number of configurations and also the
number of possible trajectories (with non-zero probability) is finite.

\section{Green function}

\label{Green function}

Consider the probability for the system to be in a configuration $\bm{x}$ after $t$ time steps, $t\geq 0$,
\begin{equation}
G_t(\bm{x}|\bm{y})= \sum_{\{\bm{x}^0,\bm{x}^1,\dots,\bm{x}^{t-1}\}}P(\bm{x}%
^0,\dots,\bm{x}^t),  \label{green}
\end{equation}
given
\begin{equation}
P_0(\bm{x}^0)=\delta_{\bm{y},\bm{x}^0}.
\end{equation}
The sum is taken over all trajectories, the final configuration being fixed
\begin{equation}
\bm{x}^t\equiv \bm{x}.
\end{equation}
The transition probability satisfies the recurrence relation
\begin{equation}
G_t(\bm{x}|\bm{y})=\sum_{\bm{x'}\in\mathbb{Z}^N_>} G_1(\bm{x}|\bm{x'}) G_{t-1}(\bm{x'}|\bm{y}),  \label{master}
\end{equation}
which is a particular case of the Chapman-Kolmogorov equation. It follows directly from the definitions
(\ref{green}) and (\ref{weights}) that
\begin{eqnarray}
G_0(\bm{x}|\bm{x^0})&=&\delta_{\bm{x},\bm{x^0}}  \label{in_cond} \\
G_1(\bm{x}|\bm{x'})&=& P_1(\bm{x}|\bm{x'}) .  \label{G_1}
\end{eqnarray}
An explicit form of $G_t(\bm{x},\bm{y})$ for arbitrary time $t$ can be found
as a Green function (GF) of the equation (\ref{master}) that satisfies (\ref%
{in_cond}), (\ref{G_1}). The solution is particularly simple in the case of the TASEP, having the form of a
determinant of an $N\times N$ matrix. A first determinantal formula for the GF $G_t(\bm{x},\bm{y})$ has been
found by Sch\"utz \cite{Sch1} for the continuous time TASEP. A similar formula for the TASEP with backward
sequential update has been obtained by the use of trajectory treatment of the Bethe Ansatz
\cite{priezzhev_traject}. As pointed out in \cite{RS} it can be proved rigorously using the determinantal
approach of \cite{Sch1}. For the sake of self-containedness and as introduction to the approach used further
below for multi-time correlations we present here a complete proof of this formula by induction.

\begin{proposition}
\label{Green_fun_equal_time} The GF of (\ref{master}) has the determinantal form
\begin{equation}
G_t(\bm{x}|\bm{y})=\det[F_{j-i}(x_i-y_j,t)]_{i,j=1,\dots,N}, \label{Green_det_equal time}
\end{equation}
where
\begin{equation}  \label{FDisc}
F_n(x,t)=\left\{{\
\begin{array}{ll}
\frac{1}{2\pi \mathrm{i}}\oint_{\Gamma_0}\frac{dw}{w} \left(q+\frac{p}{w}%
\right)^t(1-w)^{-n}w^x, & t\geq0 \\
0, & t<0%
\end{array}%
} \right.,
\end{equation}
and the integration contour $\Gamma$ encircles the origin, while the point w=1 stays outside.
\end{proposition}

\begin{proof}
The proposition can be proved in three steps. First we prove that the r.h.s.
of (\ref{Green_det_equal time}) satisfies the initial conditions (\ref%
{in_cond}), then that it satisfies the recurrence relation (\ref{master}), and finally show that it correctly
defines $G_1(\bm{x}|\bm{x'})$ according
to (\ref{G_1}), so that it can be used to obtain $G_t(\bm{x}|\bm{y})$ for $%
t>1$.

We first expand the determinant representation of $G_t(\bm{x}|\bm{y})$ given in r.h.s. of (\ref{Green_det_equal
time}) into a sum over permutations:
\begin{equation}
G_t(\bm{x}|\bm{y})=\sum_{\sigma\in S_N}(-1)^{|\sigma|} \prod_{i=1}^{N}F_{\sigma_i-i}(x_i-y_{\sigma_i},t),
\label{sumdet}
\end{equation}
where $\sigma=(\sigma_1,\dots,\sigma_N)$ is a permutation of $N$ integers $%
(1,\dots,N)$ and the summation runs over all elements of the symmetric group $S_N$. To show that
$G_t(\bm{x}|\bm{y})$ satisfies the initial conditions we need to evaluate this sum in the particular case $t=0$.
To this end, we
refer to the first statement of the Lemma \ref{cycles} proved in appendix %
\ref{cycles appendix}. It directly follows from this statement that the only summand which can be nonzero in the
r.h.s. of (\ref{sumdet}) in the case $t=0 $ is the one consisting of the diagonal elements of the matrix
$\left[F_{j-i}(x_i-y_j,0)\right]_{1\leq i,j \leq N}$. Then the initial conditions are straightforwardly verified,
\begin{equation}
G_0(\bm{x}|\bm{y})=\delta_{\bm{x},\bm{y}}.
\end{equation}
This completes the first step of the proof.

To prove the recurrence relation (\ref{master}) for the r.h.s of (\ref%
{Green_det_equal time}) we use recurrence relations \cite{RS} for the functions $F_n(x,t)$:
\begin{equation}  \label{FIden1}
F_n(x,t)=qF_n(x,t-1)+pF_n(x-1,t-1),
\end{equation}
\begin{equation}  \label{FIden2}
F_n(x+1,t)=F_{n}(x,t)-F_{n-1}(x,t).
\end{equation}
Using (\ref{FIden1}) we can expand the determinant in the r.h.s. of (\ref%
{Green_det_equal time}) into a sum of similar determinants, where we have
either $F_{i-j}(x_j-y_i,t-1)$ or $F_{i-j}(x_j-y_i-1,t-1)$ instead of $%
F_{i-j}(x_j-y_i,t)$ with coefficients $q$ and $p$ respectively. In other words, we obtain another recurrence
relation
\begin{equation}  \label{Free}
G_t(\bm{x}|\bm{y})=\sum_{\bm{x'}\in\mathbb{Z}^N_\geq}G_1^{\mathrm{free}}(%
\bm{x}|\bm{x^{'}})G_{t-1}(\bm{x^{'}}|\bm{y}),
\end{equation}
similar to (\ref{master}), but with transition probabilities for the last step written for $N$ non-interacting
particles:
\begin{equation}
G_1^{\mathrm{free}}(\bm{x}|\bm{x^{'}})=\prod_{i=1}^N F_0(x_i-x^{^{\prime }}_i,1)  \label{GreenDetFree}
\end{equation}
While the configuration $\bm{x}$ is supposed to be from $\mathbb{Z}^N_>$, given by (\ref{range}), the range of
summation in (\ref{GreenDetFree}) is the set
\begin{equation}
\mathbb{Z}^N_\geq=\left\{(x_1,\dots,x_N) \in \mathbb{Z}^N : x_{1}\geq x_{2}\geq \dots \geq x_{N}\right\},
\end{equation}
which admits coincidence of coordinates of two successive particles, $%
x^{\prime }_i=x^{\prime }_{i+1}$. Of course, the quantity $G_{t-1}(\bm{x^{'}}%
|\bm{y})$ being defined for $\bm{x'}\in \mathbb{Z}^N_>$ does not have the meaning of a probability in
$\mathbb{Z}^N_\geq$. Nevertheless we can extend formally the determinant (\ref{Green_det_equal time}) to this
set.

Now we show that (\ref{Free}) with (\ref{GreenDetFree}) is equivalent to (%
\ref{master}). Compare first  quantities $G_1^{\mathrm{free}}(\bm{x}|%
\bm{x^{'}})$ and $G_1(\bm{x}|\bm{x^{'}})$. Let us consider a summand of (\ref%
{sumdet}) corresponding to a particular permutation $\sigma$ at $t=1$. The second statement of the Lemma
\ref{cycles} restricts the range of permutations $\sigma$ as well as the range of configurations $\bm{x'}$ for
this summand to be nonzero at given $\bm{x}$. To characterize the structure of permutations, we use the notion of
cycles. Given the permutation $\sigma$ the cycle is a set of elements $(i_1,\dots,i_k)$ such that
\begin{equation}
i_2=\sigma_1,\dots,i_k=\sigma_{i_{k-1}},i_1=\sigma_{i_{k}}.
\end{equation}
Any permutation can be decomposed into a set of disjoint cycles. According to the Lemma \ref{cycles}, the
permutations contributing to the
sum can contain cycles of two types only. First, these are trivial cycles, $%
\sigma_i=i$, which imply either $x_i^{^{\prime }}=x_i$ or $x_i^{^{\prime
}}=x_i-1$. They contribute to the product the terms $F_0(0,1)=q$ and $%
F_0(1,1)=p$ respectively. Second, the cycles which consist of the indices of successive particles packed into
compact clusters:
\begin{eqnarray}
i_1&=&i,\dots,i_k=i+k-1  \label{cluster1} \\
x_i&=&x_i^{^{\prime }},\dots,x_{i+k-1}=x_{i+k-1}^{^{\prime }}=x_i-k+1 \label{cluster2}
\end{eqnarray}
for $1\leq i\leq N-1$ and $2\leq k\leq N-i+1$, so that positions of clusters are the same in $\bm{x}$ and
$\bm{x'}$. The latter  case  gives  the terms of type
\begin{eqnarray}
C_k&=&(-1)^{k-1} F_{i_2-i_1}(x_{i_1}-x^{^{\prime }}_{i_2},1)  \notag \\
&\dots& F_{i_k-i_{k-1}}(x_{i_{k-1}}-x^{^{\prime }}_{i_{k}},1) F_{i_1-i_k}(x_{i_k}-x^{^{\prime }}_{i_1},1),
\label{ClusterC}
\end{eqnarray}
which is illustrated by the diagram in Fig.\ref{cluster}a.
\begin{figure}[tbp]
\unitlength=1mm \makebox(130,60)[cc]
 {\psfig{file=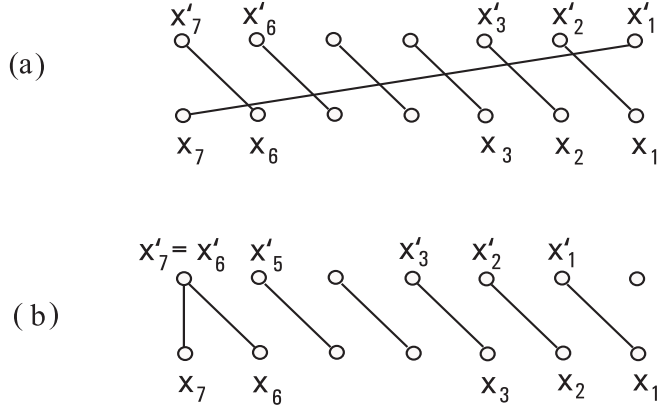,width=130mm}}
\caption{Clusters $C_k$, (a) and $D_k$, (b) for $k=7$} \label{cluster}
\end{figure}
As a result, we obtain
\begin{equation}
C_k=(-1)^{k-1} F_1(1,1)^{k-1} F_{-k+1}(-k+1,1)=p^{k-1}q.  \label{clusterC}
\end{equation}

The transition probability $G_1^{\mathrm{free}}(\bm{x}|\bm{x'})$ in (\ref%
{Free}) also contains the terms $F_0(0,1)=q$ and $F_0(1,1)=p$ corresponding
to free particles staying or making a jump with $x_i^{^{\prime }}=x_i$ or $%
x_i^{^{\prime }}=x_i-1$ respectively. In addition one must take into account cases when two particles from the
same site of $\bm{x'}$ are coming to neighboring sites of $\bm{x}$. In general they can also belong to a bigger
cluster of particles ahead making a step as well. This process, shown in Fig.%
\ref{cluster}b, contributes the term
\begin{equation}
D_k= F_0(x_{i}-x^{^{\prime }}_{i},1)\dots F_0(x_{i+k-2}-x^{^{\prime }}_{i+k-2},1) F_{0}(x_{i+k-1}-x^{^{\prime
}}_{i+k-1},1),  \label{ClusterD}
\end{equation}
where
\begin{eqnarray}
x_{i+1}&=&x_{i}-1, \dots, x_{i+k-1}=x_{i}-k+1,  \notag \\
x_{i}^{^{\prime }}&=&x_{i}-1, \dots, x_{i+k-1}^{^{\prime }}=x_{i}-k+1, \label{x'free}
\end{eqnarray}
so that
\begin{equation}
D_k= F_0(1,1)^{k-1} F_{0}(0,1)=p^{k-1}q  \label{clusterD}
\end{equation}
Thus, there is a one-to-one correspondence between $G_1^{\mathrm{free}}(%
\bm{x}|\bm{x^{'}})$ with $\bm{x^{'}}\in \mathbb{Z}_{\geq}^N$ and the nonzero
summands entering the determinant expansion (\ref{sumdet}) for $G_1(\bm{x}|%
\bm{x^{'}})$ with $\bm{x^{'}}\in \mathbb{Z}_{>}^N$.

Consider now the terms $G_{t-1}(\bm{x'}|\bm{y})$ in the r.h.s. of (\ref%
{master}) and (\ref{Free}) for the configurations $\bm{x'}$, which differ by groups of arguments defined by
(\ref{cluster2}) and (\ref{x'free}) respectively, i.e.
\begin{equation}
G_{t-1}(\dots,x_i,x_i-1,\dots,x_i-k+1,\dots|\bm{y})  \label{shiftC}
\end{equation}
and
\begin{equation}
G_{t-1}(\dots,x_{i}-1,x_{i}-2, \dots, x_{i}-k+1, x_{i}-k+1,\dots|\bm{y}). \label{shiftD}
\end{equation}
The first $k-1$ arguments in (\ref{shiftC}) are shifted by one with respect to those in (\ref{shiftD}). Using the
identity (\ref{FIden2}) and performing
column operations in (\ref{Green_det_equal time}) for columns of the matrix $%
F_{j-i}(x_i^{^{\prime }}-y_j|t-1)$ containing shifted arguments, we equalize the determinants for (\ref{shiftC})
and (\ref{shiftD}) and therefore prove that (\ref{master}) and (\ref{Free}) have the same form.

As the final step, we show that the determinant formula for $G_1(\bm{x}|\bm{x'})$ complies with the definition
(\ref{weights}) of $P_1(\bm{x}|\bm{x'})$. According to (\ref{weights}) a particle which jumps a step forward
brings a factor of $p$ into $P_1(\bm{x}|\bm{x'})$. A particle which does not jump occupies the same site in
$\bm{x}$ and in $\bm{x'}$ and can bring either the factor of $q$ or the factor of $1$. The latter occurs when
there is a particle in the right neighboring site that keeps its place as well. In other words, an isolated
compact cluster of particles keeping their positions when the system jumps from $\bm{x'}$ to $\bm{x}$ always
brings only a factor of $q$, no matter how many particles are in the cluster (by isolated cluster we mean that it
does not belong to any bigger cluster).

We need to check if the r.h.s. of (\ref{sumdet}) taken with $t=1$ gives the same result. One can see that
permutations contributing nontrivially to the sum in (\ref{sumdet}) permute only the indices of subsequent
particles constituting compact clusters, which keep their positions unchanged. Therefore the sum (\ref{sumdet})
can be factorized into factors each coming from an isolated cluster and the factors of $p$ coming from every
single
jumping particle. The latter is already in correspondence with (\ref{weights}%
). What we need is to sum up all the contributions from an isolated cluster. The result is $q$ which can be
proved by induction. Indeed, for a cluster consisting from a single isolated particle there is only a term
$F_0(0,1)=q$ corresponding to the trivial cycle. Let us suppose that it is valid for all clusters consisting of
$i<k$ particles. For $k$ particles the last particle can enter the cycle of the type
(\ref{cluster1},\ref{cluster2}), shown in the Fig.\ref{cluster}a, together with $(i-1)$, $1 \leq i < k$,
particles next to it. The contribution of such cycles is $C_i$ while the sum over all contributions of the other
$(k-i)$ particles brings the factor of $q$ by the induction hypothesis valid for $i<k$. An exception is $i=k$,
when no particles in the cluster remain anymore, resulting simply in $C_k$. Summing up all these contributions we
have
\begin{equation}
q \sum_{i=1}^{k-1}C_i +C_{k}=q,
\end{equation}
which completes the induction.
\end{proof}

\section{Generalized Green Function}

A different view at the TASEP evolution is a representation of the particle trajectories as an ensemble of
non-intersecting lattice paths. This is convenient because then we can consider a wider class of boundary
conditions, i.e., initial and final configurations which refer to different times for different paths.

\subsection{Ensembles of Lattice Paths}

Consider the two-dimensional integer lattice $L=\mathbb{Z}^2$. A point of this lattice is associated with two
integers $(x,t)$,  $x$ corresponding to a particle location and $t$ representing discrete time. Let us fix two
points $(x^0,t^0)\in L$ and $(x,t)\in L$, such that $t>t_0$.

\begin{definition}
A directed lattice path of the length $(t-t_0)$ with the starting point $%
(x^0,t^0)$ and the endpoint $(x,t)$ is a subset of $L$ of the form
\begin{equation}
\Pi_{(x,t)}^{(x^0,t^0)}=\bigcup_{i=0}^{t-t^0}\{(x^i,t^i)\}
\end{equation}
where
\begin{equation}
x^{t-t^0}\equiv x,
\end{equation}
\begin{equation}
t^i= t^0+i,
\end{equation}
and $x^i\in \mathbb{Z}$ can take any values for $i=1,\dots,t-t^0-1$.
\end{definition}

A path is associated with a weight $W\left(\Pi_{(x,t)}^{(x^0,t^0)}\right)$ defined as a product of weights of
elementary (two-point) weights,
\begin{equation}
W\left[\Pi_{(x,t)}^{(x^0,t^0)}\right]=\prod_{i=0}^{t-t^0-1} W\left[%
\left\{(x^i,t^i),(x^{i+1},t^{i+1})\right\}\right].  \label{weight}
\end{equation}
In the definition of lattice paths, the adjective "lattice" refers to the connectivity structure defined by a
given set of nonzero elementary weights, which can be represented as bonds and diagonal links on a 2D lattice
with vertices in $L$. The paths which have nonzero weight, which are the only ones of interest, can be drawn on
this lattice. In our case, the lattice paths are intended to model the stochastic trajectories of particles.
Correspondingly, the set of nonzero elementary weights is defined by the allowed particle jumps. For example, for
the 1D directed Bernoulli random walk, which is a particular case of the above dynamics with a single particle,
the nonzero elementary weights are
\begin{equation}
W\left(\left\{(x,t),(x,t+1)\right\}\right)=q,\,\,\, W\left(\left\{(x,t),(x+1,t+1)\right\}\right)=p,
\end{equation}
where $(x,t)\in \mathbb{Z}^2$.

Within these definitions, the GF of a single particle for the directed Bernoulli random walk, which is the $N=1$
case of the GF defined in the previous section, is given by the sum of the weights of all possible paths with
given starting and end points:
\begin{eqnarray}
G_t(x|x^0)=\sum_{\left\{\Pi_{(x,t)}^{(x^0,t^0)}\right\}}W\left[%
\Pi_{(x,t)}^{(x^0,t^0)}\right].
\end{eqnarray}

To interpret the results of previous section for the $N$-particle TASEP GF, we have to consider $N$
\textit{interacting} paths on the lattice. Note that the results for the GF connecting two particle
configurations at different times strictly depend on the fact that the particles are non-colliding and stay
ordered during the whole evolution. This is also true for their space-time trajectories, i.e. for the
corresponding lattice paths. Indexed at the initial time moment, the points associated with different paths have
the same order at any later time. Our aim is to define a generalization of the GF in such a way that the initial
and final positions of different particles would refer to different time moments. For this definition to be of
practical use, we would like to keep the path ordering condition. It turns out, that this is still possible if we
restrict the range of initial and final configurations:

\begin{definition}
A set of $N$ indexed points of $L$
\begin{equation*}
(\bm{x},\bm{t})=((x_1,t_1),\dots,(x_N,t_N))
\end{equation*}
is called admissible if the following conditions are satisfied simultaneously
\begin{equation}\label{x order}
x_1>x_2>\dots>x_N
\end{equation}
and
\begin{equation}\label{t order}
t_1\leq t_2 \leq \dots \leq t_N.
\end{equation}
\end{definition}

\begin{definition}
\label{N-path} Given two admissible $N$-point configurations
\begin{eqnarray}
(\bm{x^0},\bm{t^0})=\left((x_1^0,t_1^0),\dots,(x_N^0,t_N^0)\right) \\
(\bm{x},\bm{t})=\left((x_1,t_1),\dots,(x_N,t_N)\right)
\end{eqnarray}
such that
\begin{equation}
t_i> t_i^0 \,\,\,\mathrm{for} \,\,\, i=1,\dots,N,
\end{equation}
an N-path is a subset of $L\times\{1,\dots,N\}$ of the form
\begin{equation}
\bm{\Pi}_{(\bm{x,t})}^{(\bm{x^0,t^0})}=\left\{{(\Pi_1)}%
^{(x^0_1,t^0_1)}_{(x_1,t_1)}, \dots,{(\Pi_N)}_{(x_N,t_N)}^{(x^0_N,t^0_N)}%
\right\}.
\end{equation}
The number
\begin{equation}
T=t_N-t_1^0
\end{equation}
is referred to as the length of the N-path.
\end{definition}

The interaction in the TASEP implies that the statistical weight $W\left[%
\bm{\Pi}\right]$ of an $N$-path is not in general the product of the weights of $N$ single paths. Namely, it has
a form of the product of terms depending on two adjacent paths,
\begin{equation}
W\left[\bm{\Pi}\right]=W\left[\Pi_1\right]\prod_{i=1}^{N-1} W\left[{\Pi_{i+1}%
}|{\Pi_i}\right].
\end{equation}
The weight $W\left[{\Pi_{i+1}}|{\Pi_i}\right]$ is a product of two point terms for the second path, $\Pi_{i+1}$
\begin{equation}
W\left[{\Pi_{i+1}}|{\Pi_i}\right]=\prod_{j=0}^{t_{i+1}-t^0_{i+1}-1} W\left[%
\left\{(x^j_{i+1},t^j_{i+1}),(x^{j+1}_{i+1},t^{j+1}_{i+1})\right\}|\Pi_i%
\right],
\end{equation}
conditioned by the configuration of the first path $\Pi_{i}$. All the nonzero two-point terms are defined as
follows
\begin{eqnarray}
W\left[\left\{(x,t),(x+1,t+1)\right\}|\Pi\right]=\left\{
\begin{matrix}
p, & (x+1,t+1) \notin \Pi \label{weight_p} \\
0, & (x+1,t+1) \in \Pi%
\end{matrix}%
, \right. \\
W\left[\left\{(x,t),(x,t+1)\right\}|\Pi\right]=\left\{
\begin{matrix}
q, & (x+1,t+1) \notin \Pi \\
1, & (x+1,t+1) \in \Pi%
\end{matrix}%
. \right.  \label{weight_q}
\end{eqnarray}

Now we are in a position to introduce the generalized Green function (GGF).

\begin{definition}
Given two $N$-point configurations $(\bm{x,t})$ and $(\bm{x_0,t_0})$
satisfying the conditions of definition \ref{N-path}, a GGF from $(%
\bm{x_0,t_0})$ to $(\bm{x,t})$ is
\begin{equation}
G\left((\bm{x,t})|(\bm{x_0,t_0})\right)= \sum_{\left\{\bm{\Pi}_{(\bm{{x,t})}}^{%
{(\bm{x^0,t^0})}}\right\}} W\left[\bm{\Pi}_{(\bm{{x,t})}}^{{(\bm{x^0,t^0})}}%
\right]
\end{equation}
where summation is over all possible $N$-paths between $(\bm{x^0,t^0})$ and $%
(\bm{{x,t}})$.
\end{definition}

Obviously the GF $G_t\left(\bm{x}|\bm{x_0}\right)$ defined in the Section %
\ref{Green function} is a particular case of the GGF with $%
t_1^0=\dots=t_N^0\equiv t^0$ and $t_1=\dots=t_N=t+t^0$. Of course, it does not depend on $t^0$ because of the
time shift invariance. The main result of this section is the proof that the determinantal form
(\ref{Green_det_equal time}) of $G_t\left(\bm{x}|\bm{x_0}\right)$ can be extended to the GGF.

\begin{proposition}
Given two admissible \footnote{ As shown in \cite{Brankov}, the range of initial and final configurations, which
ensure the validity of the proposition, is even broader. Why we limit our consideration by the admissible
configurations only  will be clear from the subsection 4.2. }
 $N$-point configurations $(\bm{x,t})$ and $(\bm{x_0,t_0})$,
such that $t_i>t_i^0$ for $i=1,\dots,N$, the GGF can be expressed in the determinantal form:
\begin{equation}
G\left((\bm{x,t})|(\bm{x_0,t_0})\right)=\det[F_{j-i}(x_i-x_j^0,t_i-t_j^0)]%
_{i,j=1,\dots,N},  \label{GGF_det}
\end{equation}
where $F_n(x,t)$ is defined in (\ref{FDisc}).
\end{proposition}

\begin{proof}
Consider an $N$-path $\bm{\Pi}_{(\bm{{x,t})}}^{{(\bm{x^0,t^0})}}$. It follows from the definition of weights,
that if we cut an $N$-path along points of equal time, the parts obtained will be independent of each other. To
be specific, let us fix $t^{\prime }$, $t_1^0\leq t^{\prime }\leq t_N$. In general, $N$ paths forming the $N$-
path can be subdivided into three groups: ending up before or at $t^{\prime }$; starting before $t^{\prime }$
and ending later than $t^{\prime }$; starting at or later than $t^{\prime }$%
. In other words, there are two integers, $1\leq k_1\leq k_2\leq N$, such that
\begin{eqnarray}
t_1,\dots,t_{k_1}&\leq & t^{\prime } \\
t_{k_1+1},\dots,t_{k_2}&>&t^{\prime }>t_{k_1+1}^0,\dots,t_{k_2}^0 \\
t_{k_2+1}^0,\dots,t_{k_N}^0&\geq&t^{\prime }.
\end{eqnarray}

Let us fix two pairs of configurations of $k_2$ points
\begin{eqnarray}
(\bm{x}^0,\bm{t}^0)_{-}&=&\left((x_1^0,t_1^0),\dots,
(x_{k_2}^0,t_{k_2}^0)\right)  \label{(x^0,t^0)_-} \\
(\bm{x}^{\prime },\bm{t}^{\prime })_{-}&=&\left((x_1,t_1),\dots,(x_{k_1},t_{k_1}),(x^{\prime }_{k_1+1},t^{\prime
}),\dots, (x^{\prime }_{k_2},t^{\prime })\right)
\end{eqnarray}
and $(N-k_2)$ points
\begin{eqnarray}
(\bm{x}^{\prime },\bm{t}^{\prime })_{+}&=&\left((x^{\prime }_{k_1+1},t^{\prime }),\dots,(x^{\prime
}_{k_2},t^{\prime
}),(x_{k_2+1}^0,t_{k_2+1}^0), \dots,(x_N^0,t_N^0)\right) \\
(\bm{x},\bm{t})_{+}&=&\left((x_{k_1+1},t_{k_1+1}),\dots,(x_N,t_N)\right) \label{(x,t)_+}
\end{eqnarray}
where $(x^{\prime }_{k_1+1},t^{\prime }),\dots,(x^{\prime }_{k_2},t^{\prime
})$ are the intermediate points of the paths ${\Pi_{k_1+1}}$,\ $\dots$,${%
\Pi_{k_2}}$ respectively. For the weights being nonzero they have to be bounded to
\begin{equation}
x_{k_1+1}\geq x^{\prime }_{k_1+1}>\dots>x^{\prime }_{k_2} \geq x_{k_2}^0 , \label{range of x'}
\end{equation}
which ensures that the configurations (\ref{(x^0,t^0)_-})-(\ref{(x,t)_+}) are admissible. Therefore, we can
represent our $N$-path as a union
\begin{equation}
\bm{\Pi}_{(\bm{{x,t})}}^{{(\bm{x^0,t^0})}}=\bm{\widetilde{\Pi}}_{(%
\bm{{x',t'}})_-}^{{(\bm{x^0,t^0})_-}} \bigcup\bm{\widetilde{\Pi}}_{(%
\bm{{x,t}})_+}^{{(\bm{x',t'}})_+}
\end{equation}
of $k_2$-path and $(N-k_1)$-path respectively, which are clearly independent, so that  the weight of the $N$-path
is equal to the product of two weights:
\begin{equation}
W\left[\bm{\Pi}_{(\bm{{x,t})}}^{{(\bm{x^0,t^0})}}\right]= W\left[%
\bm{\widetilde{\Pi}}_{(\bm{{x',t'}})_-}^{{(\bm{x^0,t^0})_-}} \right]W\left[%
\bm{\widetilde{\Pi}}_{(\bm{{x,t}})_+}^{{(\bm{x',t'}})_+}\right].
\end{equation}
Then, we can split the summation over  $N$-paths into independent
summations over all possible realizations of $\bm{\widetilde{\Pi}}_{(%
\bm{{x',t'}})_-}^{{(\bm{x^0,t^0})_-}}$ and $\bm{\widetilde{\Pi}}_{(\bm{{x,t}}%
)_+}^{{(\bm{x',t'}})_+}$, each resulting in the corresponding GGF, and a summation over the intermediate points
$x^{\prime }_{k_1+1},\dots,x^{\prime }_{k_2}$ which vary in the range (\ref{range of x'}):
\begin{eqnarray}
G\left({(\bm{{x,t})}}|{{(\bm{x^0,t^0})}}\right)=&&  \label{GGF convolution}
\\
\sum_{\{x^{\prime }_{k_1+1},\dots,x^{\prime }_{k_2}\}}&&G\left({(\bm{{x,t}}%
)_+}|{{(\bm{x',t'}})_+}\right) G\left({(\bm{{x',t'}})_-}|{{(\bm{x^0,t^0})_-}}%
\right)  \notag
\end{eqnarray}
Using this formula and Proposition \ref{Green_fun_equal_time} we can construct the GGF recursively, attaching one
time step at a time.

\textit{First step:} Consider an $N$-path of unit length $T=t_N-t_1^0=1$. In this case all the paths constituting
the $N$-path must fall within the same unit time interval,
\begin{eqnarray}
&&t_1^0=\dots=t_N^0\equiv t^0 \\
&&t_1=\dots=t_N=t^0+1.
\end{eqnarray}
Obviously, a GGF in this case coincides with the ordinary GF discussed in Section \ref{Green function}:
\begin{equation}
G\left({(\bm{{x,t}})}|{{(\bm{x^0,t^0})}}\right)= G_1\left({\bm{x}}|{{\bm{x^0}%
}}\right),
\end{equation}
which fits (\ref{GGF_det}), thus proving the proposition for this special case.

\textit{Recursion:} Suppose (\ref{GGF_det}) is valid for the $N$-paths of
any length up to $(t-1)$. Consider $N$-paths of length $t$ and apply (\ref%
{GGF convolution}) with $t^{\prime }=t-1$,  dividing an $N$-path of length $t$ into parts of lengths $t-1$ and
the unit length. Again, the unit length part is given by the regular GF:
\begin{eqnarray}
&G\left({(\bm{{x,t})_+}}|{{(\bm{x',t'})_+}}\right)=\hspace{7cm}& \\
&G_1\left((x_{k_1+1},\dots,x_{N})|(x^{\prime }_{k_1+1},\dots,x^{\prime }_{k_2},x_{k_2+1}^0, \dots,x_N^0)\right),&
\notag
\end{eqnarray}
Then (\ref{GGF convolution}) can be rewritten as
\begin{eqnarray}
G\left({(\bm{{x,t})}}|{{(\bm{x^0,t^0})}}\right)= \sum_{\{x^{\prime
}_{k_1+1},\dots,x^{\prime }_{k_N}\}\in \mathbb{Z}^{N-k_1}_>}G\left({(%
\bm{{x',t'}})_-}|{{(\bm{x^0,t^0})_-}}\right)  \notag \\
\times G_1\left((x_{k_1+1},\dots,x_{N})|(x^{\prime }_{k_1+1}, \dots,x^{\prime }_{k_N})\right)\delta_{x^{\prime
}_{k_2+1},x_{k_2+1}^0}\dots\delta_{x^{\prime }_N,x_N^0},  \label{GGF rec1}
\end{eqnarray}
where we replaced the starting coordinates $x_{k_2+1}^0,\dots,x_N^0$ by summation variables $x^{\prime
}_{k_2+1},\dots,x^{\prime }_N$ respectively, fixing their values with the corresponding Kronecker delta symbols.

Here we can formally release the summation domain (\ref{range of x'}) to the whole $\mathbb{Z}^{N-k_1}_>$, as one
of the factors of the expression under the sum is zero beyond this domain. The first factor together with
Kronecker deltas can be packed into a single determinant in the domain (\ref{range of x'}),
\begin{eqnarray}
G\left({(\bm{{x',t'}})_-}| {{(\bm{x^0,t^0})_-}}\right)&\times&\delta_{x^{%
\prime }_{k_2+1},x_{k_2+1}^0}\dots \delta_{x^{\prime }_N,x_N^0}  \notag \\
&=& \det\left[F_{j-i}(x^{\prime }_i-x_j^0,t^{\prime }_i-t_j^0)\right]%
_{i,j=1}^N  \label{big det}
\end{eqnarray}
where we set
\begin{eqnarray}
&&x^{\prime }_1 \equiv x_1,\dots,x^{\prime }_{k_1}\equiv x_{k_1}, \\
&&t^{\prime }_1 \equiv t_1,\dots,t^{\prime }_{k_1}\equiv t_{k_1},t^{\prime }_{k_1+1}=\dots= t_{N}\equiv t^{\prime
},
\end{eqnarray}
This can be verified by noting that for
\begin{eqnarray}
1\leq i\leq k_2,\,\,\,k_2+1\leq j\leq N,
\end{eqnarray}
the following matrix elements are zero
\begin{eqnarray}
F_{j-i}(x^{\prime }_i-x_j^0,t^{\prime }_i-t_j^0)=0.  \label{F(x'-x^0)}
\end{eqnarray}
Indeed, we have either
\begin{eqnarray}
t^{\prime }_i-t_j^0<0
\end{eqnarray}
or
\begin{eqnarray}
t^{\prime }_i-t_j^0&=&0, \\
x^{\prime }_i-x_j^0&=&x^{\prime }_i-x^{\prime }_j>0,
\end{eqnarray}
as the configuration $(\bm{x',t'}%
)_+$ is admissible. Then (\ref{F(x'-x^0)}) follows from the properties of $%
F_n(x,t)$, (\ref{F_n=0}-\ref{n>0}). Thus the matrix on the r.h.s. of (\ref%
{big det}) has a block form
\begin{equation}
\left[F_{j-i}(x^{\prime }_i-x_j^0,t^{\prime }_i-t_j^0)\right]_{i,j=1}^N=
\begin{pmatrix}
\bm{A} & \bm{C} \\
\bm{0} & \bm{B}%
\end{pmatrix}%
,
\end{equation}
where
\begin{eqnarray}
\bm{A}=\left[F_{j-i}(x^{\prime }_i-x_j^0,t^{\prime }_i-t_j^0)\right]%
_{i,j=1,\dots,{k_2}} \\
\bm{B}=\left[F_{j-i}(x^{\prime }_i-x_j^0,t^{\prime }_i-t_j^0)\right]_{i,j={%
k_2}+1,\dots, N}
\end{eqnarray}
By definition
\begin{eqnarray}
\det \bm{A}=G\left({(\bm{{x',t'}})_-}|{{(\bm{x^0,t^0})_-}}\right)
\end{eqnarray}
and from Proposition \ref{Green_fun_equal_time}
\begin{eqnarray}
\det \bm{B}=G_0((x^{\prime }_{k_2+1},\dots,x^{\prime }_N)|(x_{k_2+1}^0,\notag \dots,x_N^0))\\=\delta_{x^{\prime
}_{k_2+1},x_{k_2+1}^0}\dots\delta_{x^{\prime }_N,x_N^0}.
\end{eqnarray}
Then (\ref{big det}) follows from
\begin{equation}
\det
\begin{pmatrix}
\bm{A} & \bm{C} \\
\bm{0} & \bm{B}%
\end{pmatrix}%
=\det \bm{A} \det \bm{B}.
\end{equation}

It remains to prove that a substitution of the determinant for $G_1\left(%
\bm{x}| \bm{x'}\right)$ and the r.h.s of (\ref{big det}) into (\ref{GGF rec1}%
) results in the determinant (\ref{GGF_det}). This is completely analogous
to the proof of the recurrence relation in Proposition \ref%
{Green_fun_equal_time}, up to the fact that only the matrix elements from the columns with the indices
$j=k_1+1,\dots,N$ are involved in the summation.
\end{proof}

\subsection{Probabilistic content}

The GF introduced in Section \ref{Green function} defines a probability measure on the set of particle
configurations (\ref{range}) and thus has the
natural interpretation of a conditional probability to be in configuration $%
\bm{x}$ at time step $t$ given an initial configuration $\bm{x}^0$. However, the GF is also a marginal measure
defined on the set of all TASEP realizations with exactly $t$ steps. When we sum over all possible realizations
provided the final configuration at step $t$ is fixed, we get the Green function $G_t(\bm{x}|\bm{x}^0)$.
Moreover, due to the Markov character of the process, it obviously has the same meaning also on the set of
process realizations of any longer duration of $t^{\prime }\geq t$ steps. Then, the GF still gives a probability
of a particle configuration to appear at the intermediate step $t$. Integration of the GF over a subset of
configurations, where the positions of some particles are constrained to a given domain, yields a corresponding
new marginal probability, also called correlation function. An objective of the present article is an extension
of this scheme to marginal probabilities referring to particle positions at different times. To this end, we aim
to treat the GGF as a marginal measure defined over more general sets than the set of fixed time particle
configurations. Before discussing the general scheme, we first outline it with the example of a single particle.

\subsubsection{Single particle ($N$=1)}

Consider a set of lattice points
\begin{equation}
\mathcal{B}=\{(x_i,t_i)\}_{i\in \mathbb{Z}},  \label{boundary}
\end{equation}
where for any $i\in \mathbb{Z}$ either
\begin{equation}
(x_{i+1},t_{i+1})=(x_i+1,t_i)  \label{boundarypoints}
\end{equation}
or
\begin{equation}
(x_{i+1},t_{i+1})=(x_i,t_i-1).  \label{boundarypoints_v}
\end{equation}
We refer to $\mathcal{B}$ as a boundary set implying that it can be viewed as a boundary of the set
\begin{equation}
\tilde{L}=\bigcup_{(x,t)\in \mathcal{B}} \{(x^{\prime },t^{\prime })\in L :x^{\prime }\leq x, t^{\prime }\leq t,
\}.
\end{equation}
Note that
\begin{equation}
\mathcal{B} \subset \tilde{L}.
\end{equation}

The boundary set has the following obvious property

\begin{description}
\item[Property I.]  Any directed Benoulli random walk starting inside $%
\tilde{L}$ and ending in its complement, $\tilde{L}^c \equiv L \backslash \tilde{L}$, crosses the boundary.
Furthermore, once having left $\tilde{L}$
it can never return. More specifically, consider a path $%
\Pi_{(x,t)}^{(x^0,t^0)}$ such that $(x^0,t^0)\in \tilde{L}$ and $(x,t)\in
\tilde{L}^c$. Then, the path necessarily contains exactly one point $%
(x^{\prime },t^{\prime })\in \Pi_{(x,t)}^{(x^0,t^0)}$, such that
\begin{equation}
(x^{\prime },t^{\prime }) \in \mathcal{B}
\end{equation}
and
\begin{equation}
(x^{\prime \prime },t^{\prime \prime }) \in \tilde{L}^c
\end{equation}
for any $(x^{\prime \prime },t^{\prime \prime })\in \Pi_{(x,t)}^{(x^0,t^0)}$ with $t^{\prime \prime }>t^{\prime
}$.
\end{description}

For each $(x,t)\in \mathcal{B}$ we define an exit probability
\begin{eqnarray}
P_{e}^\mathcal{B}((x,t))=\sum_{(x^{\prime },t+1)\in \tilde{L}^c}\notag
W[\{(x,t),(x^{\prime },t+1)\}]\\= \sum_{(x^{\prime },t^{\prime })\in \tilde{%
L}^c}W[\{(x,t),(x^{\prime },t^{\prime })\}].
\end{eqnarray}

A generalization of the probabilistic meaning of the GF follows from the following simple claim.

\begin{proposition}
Given $(x^0,t^0)\in \tilde{L}$, the quantity
\begin{equation}
\mathcal{G}^\mathcal{B}((x',t')|(x^0,t^0))=P_{e}^\mathcal{B}((x',t')) G((x',t')|(x^0,t^0)) \label{claim}
\end{equation}
defines a probability measure on $\mathcal{B}$.
\end{proposition}

\begin{proof}
Suppose first that $\tilde{L}$ is bounded from above in time direction, i.e. there exists a time $t$ such that
\begin{equation}
\tilde{L}\subset L_t \equiv \{(x',t')\in L: t'<t , x' \in \mathbb{Z}\}.  \label{L bound}
\end{equation}
Then any path with starting point $(x^0,t^0)\in \tilde{L}$ and end point at $%
(x,t)$, with any $x\in \mathbb{Z}$, definitely leaves $\tilde{L}$ via a unique point of $\mathcal{B}$ due to
property I.

Consider a set of paths of length $T=t-t_0$
\begin{equation}
\Omega_t^{(x^0,t^0)}=\{\Pi^{(x^0,t^0)}_{(x,t)}:{x\in\mathbb{Z}}\},
\end{equation}
with all possible endpoints. Note that the weights we consider maintain the probability conservation
\begin{equation}
\sum_{\Pi\in \Omega_t^{(x^0,t^0)}}W[\Pi]=\sum_{x\in \mathbb{Z}}G((x,t)|{(x_0,t_0)})=1. \label{G norm}
\end{equation}
Therefore weights $W[\Pi^{(x_0,t_0)}_{(x,t)}]$ define a probability measure on $\Omega_t^{(x^0,t^0)}$ assigning
probability to each particular path. We can now decompose any path from $\Omega_t^{(x^0,t^0)}$ into three parts
\begin{equation}
\Pi^{(x_0,t_0)}_{(x,t)}=\Pi^{(x_0,t_0)}_{(x',t')}\bigcup\{(x',t'), (x'',t'+1)\}\bigcup\Pi^{(x'',t'+1)}_{(x,t)},
\label{decompose}
\end{equation}
where $(x^{\prime },t^{\prime })\in \mathcal{B}$ and $(x^{\prime \prime },t^{\prime }+1)\in \tilde{L}^c$. The
Markov property implies that the weight of a path is a product of the weights of its parts :
\begin{equation}
W\left(\Pi^{(x_0,t_0)}_{(x',t')}\bigcup\Pi^{(x',t')}_{(x,t)}\right)
=W\left(\Pi^{(x_0,t_0)}_{(x',t')}\right)W\left(\Pi^{(x',t')}_{(x,t)}\right).\label{Markov1}
\end{equation}

Using this property we can sum the weight of (\ref{decompose}) over all
paths from $\Omega_t^{(x^0,t^0)}$ with $(x^{\prime },t^{\prime })\in\mathcal{%
B}$ fixed and $(x^{\prime \prime },t^{\prime }+1)$ varying in $\tilde{L}^c$:
\begin{equation}  \label{curlygreen}
   \mathcal{G}^\mathcal{B}((x',t')|(x^0,t^0))=\sum_{\tiny \left\{\Pi\in \Omega_{t}^{(x^0,t^0)}:
{(x',t')\in \mathcal{B};(x'',t'+1)\in{\tilde{L}^c}}\right\}}W((\mathrm {eq. \ref{decompose}}))
\end{equation}
The summation over the last segment yields 1 due to (\ref{G norm}), while the rest results in the r.h.s. of
(\ref{claim}).

Therefore, $\mathcal{G}^\mathcal{B}((x',t')|(x^0,t^0))$ gives a
marginal probability for a path form $\Omega_t^{(x_0,t_0)}$ to exit from $%
\tilde{L}$ at a particular point $(x^{\prime },t^{\prime })$ of the boundary. As any path from
$\Omega_t^{(x_0,t_0)}$ definitely crosses the boundary, this in particular implies
\begin{equation}
\sum_{(x',t')\in \mathcal{B}}\mathcal{G}^\mathcal{B}((x',t')|
(x^0,t^0))=\sum_{\{\Pi\in\Omega_t^{(x_0,t_0)}\}}W[\Pi]=1.
\end{equation}
Now one can release the condition $(\ref{L bound})$ which ensures that every path from $\Omega_t^{(x_0,t_0)}$
crosses the boundary before time $t$. This may not happen if the boundary approaches infinite time at some finite
value of the coordinate $x>x^0$, which is to say that
\begin{equation}
\mathcal{B}\supset \bigcup_{i\in \mathbb{N}}(x,t^*+i)
\end{equation}
for some $t^*$. Then, for any $t$ there exist paths in $%
\Omega^{(x_0,t_0)}_{t}$ which do not reach the boundary within $(t-t_0)$
time steps. A simple estimate, however, shows that their total measure in $%
\Omega^{(x_0,t_0)}_{t}$ vanishes as $t\to\infty$.
\begin{equation}
\mathrm{Prob}\left(\Pi\subset\Omega^{(x_0,t_0)}_{t}: \Pi\bigcap\mathcal{B}%
=\emptyset\right)= O\left(p^{x-x^0}q^{t-t^0}\right).
\end{equation}
Therefore we have
\begin{eqnarray}
 \sum_{(x',t')\in\mathcal{B}} \mathcal{G}^\mathcal{B}((x',t')|(x^0,t^0))\equiv \lim_{t\to \infty}
\sum_{(x',t')\in\mathcal{B}\bigcap L_{t}}\mathcal{G}^\mathcal{B}((x',t')|(x^0,t^0))\notag\\
 =\lim_{t\to \infty}\sum_{\left\{\Pi\subset\Omega^{(x_0,t_0)}_{t}:
\Pi\bigcap\mathcal{B}\neq\emptyset\right\}}W[\Pi]
 =1.
\end{eqnarray}
This proves the proposition.
\end{proof}

The restriction of the set of paths to $\Omega^{(x_0,t_0)}_{t}$ allows us to interpret the function
$\mathcal{G}^\mathcal{B}((x,t)|(x^0,t^0))$ in terms of the particle evolution. Specifically,
$\Omega^{(x_0,t_0)}_{t}$ can be treated as the set of all possible sequences of $(t-t^0)$ positions of a
particle that started its walk at site $x^0$. Correspondingly, $\mathcal{G}^%
\mathcal{B}((x,t)|(x^0,t^0))$ gives a marginal probability over this set for the sequence, in which space-time
position of the particle leaves the domain $\tilde{L}$ right after taking the value $(x,t)\in \mathcal{B}$. We
may
conclude that when $(x,t+1)\notin \tilde{L}$, $\mathcal{G}^\mathcal{B}%
((x,t)|(x^0,t^0))$ gives the probability for a particle to be in site $x$ at
time $t$, irrespectivly of what the next step is. When $(x,t+1)\in \tilde{L}$%
, this is a probability for the particle to jump out of site $x$ at time $t$.

The decomposition (\ref{decompose}) and the Markov property (\ref{Markov1}) allow us to obtain a formula
\begin{eqnarray}  \label{conv}
G((x,t)|(x^0,t^0))=\sum_{{\tiny (x',t')\in \mathcal{B};(x'',t'+1)\in{\tilde{L}^c}}} G((x,t)| (x'',t'+1))\notag\\
\times  W(\{(x',t'),(x'',t'+1)\})G((x',t')|(x^0,t^0)).
\end{eqnarray}
for the convolution of the GGF associated with a general boundary which generalizes the formula (\ref{GGF
convolution}) we used for convolution along fixed time configurations.

We consider two examples of boundaries and corresponding Green functions we deal with below:

\begin{example}
Fixed time boundary,
\begin{equation}
\mathcal{B}^{(\cdot,t)}=\{(x,t):x\in\mathbb{Z}\}
\end{equation}
with $t\in\mathbb{Z}$ being fixed. This is a case already discussed in section (\ref{Green function}). Its Green
function is given by
\begin{equation}
\mathcal{G}^{\mathcal{B}^{(%
\cdot,t)}}((x,t)|(x^0,t^0))=G((x,t)|(x^0,t^0))=p^{x-x^0}q^{t- t^0-x+x^0}
\begin{pmatrix}
t-t_0 \\
x-x^0%
\end{pmatrix}
\label{G1}
\end{equation}
if $0\leq (x-x^0)\leq(t-t^0)$ and $G((x,t)|(x^0,t^0))=0$ otherwise.
\end{example}

\begin{example}
Fixed space boundary,
\begin{equation}
\mathcal{B}^{(x,\cdot)}=\{(x,t):t\in\mathbb{Z}\}  \label{fixed x}
\end{equation}
with $x>x_0\in\mathbb{Z}$ being fixed.  At each point the GF $\mathcal{G}^{%
\mathcal{B}^{(\cdot,t)}}((x,t)|(x^0,t^0))$ has to be supplied with a "jump off the boundary" probability $p$:
\begin{equation}
\mathcal{G}^{\mathcal{B}^{(x,\cdot)}}((x,t)|(x^0,t^0))=p G((x,t)|(x^0,t^0)),
\end{equation}
with $G((x,t)|(x^0,t^0))$ like in (\ref{G1}). Then the normalization and convolution read:
\begin{eqnarray}
\sum_{t=x-x^0+t^0}^{\infty} p^{x-x^0+1}q^{t-t^0-x+x^0}
\begin{pmatrix}
t-t^0 \\
x-x^0%
\end{pmatrix}
=1,
\end{eqnarray}
\begin{eqnarray}
\sum_{t'=t^{0}+x^{\prime }-x_0}^{t-x+x^{\prime }}
\begin{pmatrix}
t^{\prime}-t^0 \\
x^{\prime}-x^0%
\end{pmatrix}
\begin{pmatrix}
t-t^{\prime }-1 \\
x-x^{\prime }-1%
\end{pmatrix}
=%
\begin{pmatrix}
t-t^0 \\
x-x^0%
\end{pmatrix}%
,\,\,\,(x^0\leq x^{\prime }\leq x).
\end{eqnarray}
\end{example}

\subsubsection{$N$ particles}

We generalize now the concept of a boundary set to the $N$-paths. To motivate our approach we first observe the
following. A natural candidate for an $N$-boundary would be an arbitrary collection of $N$ boundaries of the form
(\ref{boundary}).
\begin{equation}
\bm{\mathcal{B}}=\{ \mathcal{B}_1,\dots,\mathcal{B}_N\}
\end{equation}
However, an attempt to do such a straightforward generalization fails. There
are two obstacles. First we note that the factorization of the weight of an $%
N$-path into the product of weights of its parts does not take place for arbitrary endpoints of these parts.
Indeed, consider two adjacent steps of particles $i$ and $i-1$ belonging to the $N$-path $\Pi$
\begin{equation}
\bm{\Pi}\supset\{(x,t),(x,t+1)\}_i\bigcup\{(x+1,t),(x+1,t+1)\}_{i-1}
\end{equation}
Within $\bm{\Pi}$ the weight associated with $\{(x,t),(x,t+1)\}_i$ is
\begin{equation}
W\left(\{(x,t),(x,t+1)\}|\{(x+1,t),(x+1,t+1)\}\right)=1.
\end{equation}
On the other hand if we separate a part $\tilde{\bm{\Pi}}$ of $\bm{\Pi}$, $%
\tilde{\bm{\Pi}}\subset \bm{\Pi}$, so that
\begin{equation}
\tilde{\bm{\Pi}}\supset \{(x,t),(x,t+1)\}_i, \tilde{\bm{\Pi}}%
\not\supset\{(x+1,t), (x+1,t+1)\}_{i-1},
\end{equation}
a similar weight within $\tilde{\bm{\Pi}}$ will be
\begin{equation}
W\left(\{(x,t),(x,t+1)\}_i\right)=q.
\end{equation}
Thus, if we use the decomposition separating $\tilde{\bm{\Pi}}$ from $%
\bm{\Pi}$, the weight is not factorized and hence we cannot use a convolution formula like (\ref{GGF
convolution}). In general, the difficulty
occurs when the endpoints of segments $\{(x,t),(x,t+1)\}_i$ and $%
\{(x+1,t),(x+1,t+1)\}_{i-1}$ occur in the sets $\tilde{L}_{i}$ and $\tilde{L}%
^c_{i-1}$ respectively. While they interact in the common $N$-path,  their weights become different when they do
not "see each other" in different subsets.

The second difficulty is that the GGF we dealt with is defined only for the case of $N$- paths with admissible
configurations at the endpoints. At the
same time, mutual positions of the endpoints of different paths within an $N$%
-path varying in arbitrary boundaries can violate obviously the admissibility constraint.

We need to find such a form of an $N$-boundary that it would ensure (a)  the weight factorizability when the
$N$-path is divided along $\bm{\mathcal{B}}$ and (b)  admissibility of initial and final configurations for any
$N$-path with nonzero weight. Both requirements can be fulfilled with $N$-boundaries constructed as follows. We
introduce a translation operator $\bm{T_{n}}$ translating a set of space-time points by $n$ steps in $x$
direction.
\begin{equation}
T_{n}\bigcup_i\{(x_i,t_i)\}=\bigcup_i\{(x_i+n,t_i)\}
\end{equation}

Now take an arbitrary boundary $\mathcal{B}$ of the form (\ref{boundary}). The $N$-boundary we are looking for
can be constructed from $\mathcal{B}$ as a collection of $N$ copies of $\mathcal{B}$ each translated one step
back with respect to the previous one:
\begin{equation}  \label{N-boundary}
\bm{\mathcal{B}}=(\mathcal{B},T_{-1}\mathcal{B},\dots,T_{-N+1}\mathcal{B})
\end{equation}
Indeed, this construction excludes the possibility to have the endpoints of segments $\{(x,t),(x,t+1)\}_i$ and
$\{(x+1,t),(x+1,t+1)\}_{i-1}$ in the sets $\tilde{L}_{i}$ and $\tilde{L}^c_{i-1}$ simultaneously. In addition, as
soon as path $\Pi_i$ leaves $\mathcal{B}_i$ at $(x_i,t_i)$, the coordinates of point $(x_{i+1},t_{i+1})$ of
$\mathcal{B}_{i+1}$ where path $\Pi_{i+1}$ leaves $\mathcal{B}_{i+1}$ obey $t_{i+1}\geq t_i$ and $x_{i+1}<x_i$
due to the noncrossing conditions (\ref{weight_p}), (\ref{weight_q}). The latter inequalities are the
admissibility conditions.

Using these properties we can straightforwardly generalize the one-particle scheme to $N$ particles, replacing
all the arguments above to its bold font counterparts. As a result, we have the following statement which
provides us with a probabilistic interpretation of the GGF.

\begin{proposition}
\label{GGF prob} Let $\bm{\mathcal{B}}$ be given by (\ref{N-boundary}). A quantity
\begin{equation}
\mathcal{G}^{\bm{\mathcal{B}}}((\bm{x},\bm{t})|(\bm{x}^0, \bm{t}^0))=P_{e}^{%
\bm{\mathcal{B}}}((\bm{x},\bm{t})) G((\bm{x},\bm{t})|(\bm{x}^0,\bm{t}^0))
\end{equation}
defines a probability measure on the set of admissible configurations on $%
\bm{\mathcal{B}}$, and the GGF satisfies the following recurrent relation.
\begin{eqnarray}  \label{Nconv}
G((\bm{x},\bm{t})|(\bm{x}^0,\bm{t}^0))&=&\sum_{(\bm{x}^{\prime },\bm{t}%
^{\prime })\in \bm{\mathcal{B}}; (\bm{x}^{\prime \prime },\bm{t}^{\prime
\prime })\in{\tilde{L}^c}} G(\bm{x},\bm{t}|\bm{x}^{\prime \prime },\bm{t}%
^{\prime \prime })  \notag \\
&\times& W(\{(\bm{x}^{\prime },\bm{t}^{\prime }),(\bm{x}^{\prime \prime },%
\bm{t}^{\prime \prime })\})G(\bm{x}^{\prime },\bm{t}^{\prime }|\bm{x}^0,%
\bm{t}^0).
\end{eqnarray}
\end{proposition}

\section{Signed determinantal process}

The notion of GGF allows us to consider the TASEP realizations, where different particles enter and end up their
evolution at different moments of time, such that their initial and final configurations are admissible and the
latter belong to an $N$-boundary of the form (\ref{N-boundary}). The aim of this section is to show that for specific initial conditions the measure over the realizations of these processes defined by the GGF according to Proposition \ref{GGF prob} can be represented as the conditional measure of certain signed determinantal point process, which allows for the calculation of the multi-point correlation functions. Below we give one example of this calculation for  the TASEP evolution of $N$ particles starting at the same moment of time $t=0$ at  sites $x_i^0=1-i, i=1,2,\dots,N$.

Consider a set of the $N$-paths starting at $((0,0),\dots,(1-N,0))$ and terminating by making a step away from $N$-boundary $\bm{\mathcal{B}}$ constructed as $N$ copies of the fixed $x$ boundary (\ref{fixed x}%
):
\begin{equation}  \label{bound}
\bm{\mathcal{B}}=\bigcup_{i=1}^N\{(x-i+N,t_i):t_i\in\mathbb{Z}%
\}_i
\end{equation}
where the subscript $i$ refers to a boundary set associated with a path $%
\Pi_i$ within the $N$-path corresponding to $i$-th particle. We are
interested in the joint distribution
\begin{equation}  \label{joint distrib}
\mathbf{P}=Prob\left(\{t_{n_1}\leq a_1\}\bigcap \{t_{n_2}\leq a_2\}\bigcap \dots \bigcap \{t_{n_m}\leq
a_m\}\right)
\end{equation}
of times $t_{n_1},\dots,t_{n_m}$, when $m\leq N$ particles labeled by indices $n_1<n_2<\dots<n_m$ jump off
positions $x_{n_1},\dots x_{n_m}$ taken from an array $x_i=x+N-i, i=1,\dots,N $ and $t_{n_1} \leq t_{n_2} \leq
\dots \leq t_{n_m}$. Similar calculations for other cases of a wider class of initial and final configurations,
which also can be performed by the GGF technique, will be done elsewhere.

Consider an auxiliary signed point process over the subsets of $\mathbb{Z}%
_{\geq x}\times\{1,\dots,n\}$ of the form
\begin{equation}
\mathcal{T}=\bigcup_{1\leq n \leq N } \{\tau^n_n,<\tau^n_{n-1},<\dots,<\tau^n_{1}\} \subset\mathbb{Z}_{\geq
(x-1)}\times\{1,\dots,n\}  \label{T}
\end{equation}
given by the measure
\begin{eqnarray}  \label{det_measure}
\mathcal{F}\left(\mathcal{T}\right)=\frac{1}{Z_N} \prod_{n=0}^{N-1} \det[%
\phi_{n}(\tau^{n}_{i},\tau^{n+1}_{j})]_{i,j=1}^{ n+1} \det\left[%
\Psi_i^N(\tau_{N-i}^N)\right]_{i,j=0}^{N-1},
\end{eqnarray}
where we define the functions
\begin{equation}
\phi_n(z,y)=\left\{{\
\begin{array}{c}
p,\ y \geq z \\
0,\ y < z%
\end{array}%
} \right.  \label{phi}
\end{equation}
and
\begin{equation}  \label{Psi_N}
\Psi^N_k(t)=(-1)^k \widetilde{F}_{-k}(x+N-k-1,t),
\end{equation}
where
\begin{equation}  \label{tilde{F}_n}
\widetilde{F}_n(x,t)= \frac{1}{2\pi \mathrm{i}}\oint_{\Gamma_0}\frac{dw}{w}
\left(q+\frac{p}{w}\right)^t(1-w)^{-n}w^x,
\end{equation}
The integral representation holds for $t\in\mathbbm{Z}$. This is unlike $%
F_n(x,t)$, which coincides with $\widetilde{F}_n(x,t)$, when $t\geq 0$, and vanishes at $t<0$, see (\ref{FDisc}).

We also introduce fictitious variables $\tau_{n}^{n-1}, 1\leq n\leq N$, which are effectively less than any
$\tau^{n}_j$, so that
\begin{equation}
\phi_n(\tau_{n+1}^{n},\tau^{n+1}_j)\equiv p
\end{equation}
for $j=1,..,n+1$. If we consider $\tau^n_j$ as coordinates of fictitious particles at the $n$-th time step, then
$\tau_{n+1}^{n}$ corresponds to a particle entering into the system from a reservoir on the left \cite{Sasamoto}.

Having defined the measure $\mathcal{M}\left(\mathcal{T}\right)$ we are able to interpret the GGF in the
following useful way.

\begin{proposition}
\label{GGF-measure} Given
\begin{eqnarray}
(\bm{x}^0,\bm{t}^0)&=&((0,0),\dots,(-N+1,0))  \label{in_conf} \\
(\bm{x},\bm{t})&=&((x+N-1,t_1),\dots,(x,t_N))  \label{fin_conf}
\end{eqnarray}
with $t_1\leq t_2\leq\dots\leq t_N\in\mathbb{Z}_{\geq x}$ and $x \geq -N+1$ the GGF associated with the boundary
(\ref{bound}) is a marginal of the measure $\mathcal{M}$
\begin{equation}  \label{marginal}
\mathcal{G}^{\bm{\mathcal{B}}}((\bm{x},\bm{t})|(\bm{x^0},\bm{t^0}))\equiv
p^N G((\bm{x}, \bm{t})|(\bm{x^0},\bm{t^0}))=\mathcal{M}\left(\bigcup_{k=1}^N%
\{\tau_1^k=t_k\}\right)
\end{equation}
\end{proposition}

To prove this statement one represents the GGF as a sum over the auxiliary
time variables in a way similar to that used for space variables in \cite%
{Sasamoto,BorodinFerrari,BFPS,BFS} and for time variables in \cite{Nagao}. To this end we first state and prove
the following lemma.

\begin{lemma}
\label{sasamoto sum} Given initial and final configurations (\ref{in_conf})
and (\ref{fin_conf}) respectively, the GGF $G((\bm{x},\bm{t})|(\bm{x^0},%
\bm{t^0}))$ can be represented as a sum:
\begin{eqnarray}
G((\bm{x},\bm{t})|(\bm{x}^0,\bm{t}^0))&=&(-p)^{\frac{N(N-1)}{2}}\notag \\ \times \sum_{ D}&\det&\!\!\!\!\![%
\widetilde{F}_{-N+1+i}(x+i,\tau_{j+1}^N)]_{i,j=0,\dots,N-1}  \label{2}
\end{eqnarray}
where $\tau_1^j=t_{j},\;j=1,\dots,N,$ and the summation variables take their values in the domain
\begin{equation}
D = \{\tau^j_i\in\mathbb{Z}_{\geq x},2\leq i\leq j \leq N | \tau^j_i\geq \tau^{j-1}_i, \tau^j_i >
\tau^{j+1}_{i+1}\}.  \label{3}
\end{equation}
\end{lemma}

\begin{proof}
Let $x(a,b)$ be the summation variable with its upper and lower limits $a,b$%
. Then the order of summation in (\ref{2}) is given by the sequence
\begin{eqnarray}
&&\hspace{-0.8cm}\tau^N_N(\tau^{N-1}_{N-1}-1,x),\tau^N_{N-1}(%
\tau^{N-1}_{N-2}-1,
\tau^{N-1}_{N-1}),\tau^{N-1}_{N-1}(\tau^{N-2}_{N-2}-1,x),\dots,  \notag \\
&&\hspace{-0.8cm}\tau^N_{N-2}(\tau^{N-1}_{N-3}-1,\tau^{N-1}_{N-2}),%
\tau^{N-1}_{N-2}
(\tau^{N-2}_{N-3}-1,\tau^{N-2}_{N-2}),\tau^{N-2}_{N-2}(%
\tau^{N-3}_{N-3}-1,x),\dots,  \notag \\
&&\hspace{-0.8cm}\tau^N_k(\tau^{N-1}_{k-1}-1,\tau^{N-1}_k),\dots,
\tau^k_k(\tau^{k-1}_{k-1}-1,x),\dots,  \notag \\
&&\hspace{4.5cm}\tau^N_2(\tau^{N-1}_{1}-1,\tau^{N-1}_2), \dots, \tau^2_2(\tau^{1}_{1}-1,x).  \notag
\end{eqnarray}
The variables $\{\tau^n_i,1\leq i \leq n \leq N \}$ and inequalities between them are shown in Fig. \ref{fig1}
for the case $N=4$. The total number of summations is $N(N-1)/2$.

\begin{figure}[tbp]
\unitlength=1mm \makebox(110,120)[cc] {\psfig{file=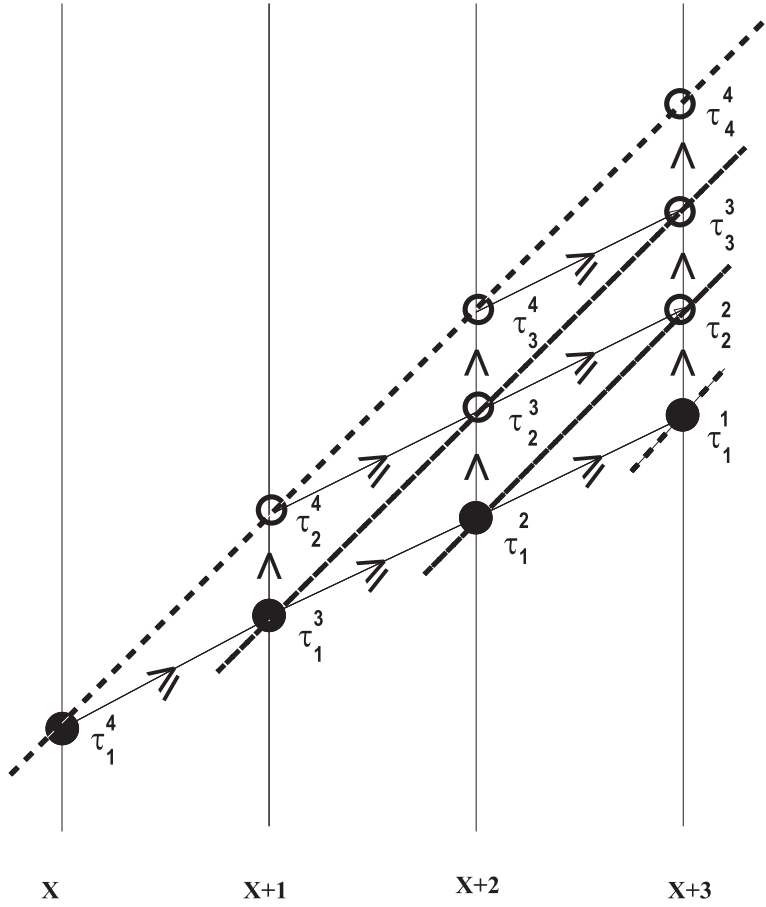,width=100mm}} \caption{The set of auxiliary variables
for $N=4$. Open circles are
integration variables, closed circles are fixed moments of time $\protect\tau%
^4_1=t_4$, $\protect\tau^3_1=t_3$, $\protect\tau^2_1=t_2$, $\protect\tau%
^1_1=t_1$. Arrows show inequalities between variables. Red lines correspond to variables with a fixed upper index
and can be considered as discrete-time steps of an auxiliary process.} \label{fig1}
\end{figure}

The summation formula for the functions $\widetilde{F}_n(x,t)$ is
\begin{equation}
p\sum_{t=t_1}^{t_2}\widetilde{F}_n(x,t) = \widetilde{F}_{n+1}(x+1,t_2+1)-%
\widetilde{F}_{n+1}(x+1,t_1)  \label{5}
\end{equation}
The first summation with $\tau^N_N$ converts the last column in the
determinant from $\widetilde{F}_{-N+1+i}(x+i,\tau^N_N)$ into $\widetilde{F}%
_{-N+2+i}(x+1+i,\tau^{N-1}_{N-1})$, $i=0,1,\dots,N-1$. Note that generally
the summation (\ref{5}) produces two terms. However the second one, $%
\widetilde{F}_{-N+2+i}(x+1+i,x)$, always vanishes. The second summation with
$\tau^N_{N-1}$ produces two terms $\widetilde{F}_{-N+2+i}(x+1+i,%
\tau^{N-1}_{N-2})-\widetilde{F}_{-N+2+i}(x+1+i, \tau^{N-1}_{N-1})$ in the previous column. The second term
coincides with the last column and disappears. The third summation with $\tau^{N-1}_{N-1}$ converts the last
column into $\widetilde{F}_{-N+3+i}(x+2+i,\tau^{N-2}_{N-2})$. Continuing, we notice that each row of summations
with $\tau^{(*)}_k$ for fixed $k$ increases the index $n$ in $\widetilde{F}_n(x,t)$ by one in each column
starting from $(N+1-k)$-th one and increases argument $x$ by one in these columns. Performing all summations, we
obtain the determinant
\begin{equation}
\det \left\vert
\begin{array}{ccc}
\widetilde{F}_{-N+1}(x,\tau^N_1) & \dots & \widetilde{F}_0(x+N-1,\tau^1_1)
\\
\vdots & \ddots & \vdots \\
\widetilde{F}_0(x+N-1,\tau^N_1) & \dots & \widetilde{F}_{N-1}(x+2N-2,%
\tau^1_1)%
\end{array}%
\right\vert .
\end{equation}%
As the variables $\tau_1^j$ are nonnegative by assumption for $j=1,\dots,N$,
we can omit the tildes, to obtain the determinant coinciding with $G((\bm{x},%
\bm{t})|(\bm{x}^0, \bm{t}^0))$ up to a factor $(-1)^{[N/2]}$ after reversing of the order of columns.
\end{proof}

Now we are in a position to prove the previous proposition.

\begin{proof}
\textit{Proof of Proposition \ref{GGF-measure}} As is illustrated in Fig.\ref%
{fig1}, the inequalities in (\ref{T}) follow directly from the inequalities
in (\ref{3}). Then, the statement of the proposition follows from Lemma \ref%
{sasamoto sum} and the fact that
\begin{equation}
\prod_{n=0}^{N-1}\det[\phi_{n}(\tau^{n}_{i},\tau^{n+1}_{j})]_{1\leq i,j\leq
n+1}=\left\{%
\begin{matrix}
p^{N(N+1)/2}, & \mathcal{T}\in D \\
0, & \mathcal{T}\notin D%
\end{matrix}%
\right.
\end{equation}
The latter can be expressed in the language of symmetric functions (see \cite%
{Macdonald} for notations and the facts used). Specifically, up to the coefficient of $p$ the function
$\phi_n(x,y)$ is nothing but the complete
symmetric function specialized at single parameter equal to 1, $%
\phi_n(x,y)=p\, h_{y-x}(1)$. Then, by the generalized Jacobi-Trudy formula
\begin{equation}  \label{Jacoby Trudy}
\det[\phi_{n}(\tau^{n}_{i},\tau^{n+1}_{j})]_{1\leq i,j\leq n+1}=p^{(n+1)} s_{\lambda^{(n+1)}/\lambda^{(n)}}(1),
\end{equation}
where we defined partitions
\begin{equation*}
\lambda^{(n)}= \tau^n_1+1-x,\tau^n_2+2-x\dots,\tau^n_n+n-x,0,\dots
\end{equation*}
for $n=1,\dots,N$. Here $s_{\lambda^{(n)}/\lambda^{(n-1)}}(x)$ is the skew
Schur  function. For $x=1$, it is equal $1$ provided that $%
\lambda^{(n)}\supset  \lambda^{(n-1)}$ and the skew partition  $%
(\lambda^{(n)}-\lambda^{(n-1)})$ is a horizontal strip. If at least one of
these conditions is violated, $s_{\lambda^{(n)}/\lambda^{(n-1)}}(1)=0$  \cite%
{Macdonald}. Translated back to the variables $\tau_i^j$, these two conditions exactly coincide with the
inequalities (\ref{3}) defining the domain $D$.
\end{proof}

The measure $\mathcal{M}(\mathcal{T})$ has the same functional form as in \cite{Sasamoto,BorodinRains,BFPS}. In
particular, Lemma 3.4 of \cite{BFPS} can be directly applied here. There is a difference in the form of functions
$\Psi_n^N(t)$ and in the space where the variables $\tau_j^i$ live, which is $%
\mathbb{Z}_{\geq x}$ rather than $\mathbb{Z}$. This difference does not affect the applicability of the Lemma,
which is formulated in a rather abstract fashion, though has to be taken into account when obtaining the final
expressions. According to Lemma 3.4 the multi-point correlation functions of $\mathcal{M}\mathcal{(T)}$ are
determinantal. To write down the kernel we introduce
\begin{equation}  \label{sec_phi}
\phi^{(n_1,n_2)}(x,y)=\left\{%
\begin{array}{ll}
(\phi_{n_1}*\phi_{n_1+1}*\dots*\phi_{n_2-1}) (x,y), & n_1<n_2 \\
0, & n_1\geq n_2%
\end{array}%
\right.,
\end{equation}
where $(a*b)(x,y)=\sum_{z\in\mathbf{Z}_{\geq x}}a(x,z)b(z,y)$, and
\begin{equation}  \label{Psi_n}
\Psi^n_{n-j}(\tau)=(\phi^{n,N}*\Psi^N_{N-j})(\tau).
\end{equation}

Consider functions
\begin{equation}  \label{space}
\{(\phi_0*\phi^{(1,n)})(\tau^0_1,\tau),\dots,(\phi_{n-2}*\phi^{n-1,n})(%
\tau^{n-2}_{n-1},\tau), \phi_{n-1}(\tau^{n-1}_n,\tau)\}.
\end{equation}
They are linearly independent and hence can serve as a basis of an $n$%
-dimensional linear space $V_n$. We construct another basis of $V_n$, $%
\{\Phi^n_j(\tau),j=0, \dots,n-1\}$, which is fixed by the orthogonality relations
\begin{equation}  \label{orthogonality}
\sum_{\tau\in\mathbb{Z}_{\geq x}}\Phi^n_i(\tau)\Psi^n_j(\tau)=\delta_{i,j}.
\end{equation}
Then, under the
\begin{description}
\item[Assumption (A)]: $\phi_n(\tau_{n+1}^n,\tau)=c_n \Phi_0^n(\tau)$
with some $ c_n\neq0$, $n=1,\dots,N$,
\end{description}
the kernel has the form
\begin{equation}  \label{Prop4.3.2}
K(n_1,\tau_1;n_2,\tau_2)=-\phi^{(n_1,n_2)}(\tau_1,\tau_2)+\sum_{k=1}^{n_2}
\Psi_{n_1-k}^{n_1}(\tau_1)\Phi_{n_2-k}^{n_2}(\tau_2).
\end{equation}
Let us now make this construction explicit.

\begin{lemma}
The functions $\Psi^n_j(\tau)$, $\Phi^n_j(\tau)$ have the following integral representation.
\begin{eqnarray}  \label{Phi_int}
\Psi^n_k(\tau)&=&\oint_{\Gamma_{0,1}}\frac{dw}{2\pi i}\left(1-p\frac{w-1}{w}%
\right)^{\tau} (w-1)^kw^{x+N-k-2} \\
\Phi^n_j(\tau)&=&p \oint_{\Gamma_1}\frac{dv}{2\pi i} \left(1-p\frac{v-1}{v}%
\right)^{-\tau-1} (v-1)^{-j-1}v^{j-N-x-1}\label{Psi_int}
\end{eqnarray}
The contours of integration $\Gamma_{0,1},\Gamma_{1}$ encircle the poles $%
w=0,1$ and $v=1$ anticlockwise respectively, leaving all the other singularities outside.
\end{lemma}

\begin{proof}
To construct $\Psi^n_k(\tau)$ we start with $\Psi^N_k(\tau)$, $k\geq 0$,
which is given by (\ref{Psi_N}), and follow recursively the definition (\ref%
{Psi_n})
\begin{equation}  \label{Psirec}
\Psi^{i-1}_{k-1}(\tau)=p\sum_{t=\tau}^\infty \Psi^i_k(t).
\end{equation}
Note that for $k\geq 0$ the expression under the integral (\ref{Psi_int}) for $\Psi^N_k(\tau)$ has no pole at
$w=1$, and therefore the integral
complies with the definition (\ref{FDisc}) of the function $\widetilde{F}%
_n(x,\tau)$ and the definition (\ref{Psi_N}) of $\Psi^N_k(t)$. Then, substituting (\ref{Psi_int}) into
(\ref{Psirec}), we can interchange integration and summation. To this end, we must ensure that the resulting sum
converges uniformly at the contour of integration. This is satisfied if
\begin{equation}
\left|1-p\frac{w-1}{w}\right| <1.
\end{equation}
The latter is true for a sufficiently large circle, $|w|\gg 1$, where this expression is arbitrarily close to
$(1-p)$. Summing the geometric
progression, we obtain the integral representation (\ref{Psi_int}) for $n<N$%
. When $k$ becomes negative, keeping the contour large ensures that the pole $w=1$ turns out to be inside the
contour together with $w=0$. At the moment we also have freedom of putting the pole $w=p/(p-1)$ inside or outside
the contour. As this is not a pole for positive $\tau$, it does not affect the convergence, though changes the
part of the sum with $\tau<0$. As it will be seen from the proof of orthogonality, it should be left outside the
contour.

Now we are in a position to define the functions $\Phi_k^n(\tau)$. First we note that
\begin{equation}
V_n=\mathrm{Span}\{\tau^{n-1},\dots,1\}
\end{equation}
and hence can be generated by any $n$ polynomials of powers from zero to $%
(n-1)$. Indeed the function $\Phi_j^n(\tau)$ defined by (\ref{Phi_int}) is a
polynomial of power $j$. To check the orthogonality relation (\ref%
{orthogonality}) we substitute the integral representations (\ref{Psi_int}, %
\ref{Phi_int}) into (\ref{orthogonality}) and interchange the order of
summation and integrations. Note that choosing  the contours as $%
\Gamma_{0,1}=\{|w|=R\}$ $\Gamma_1=\{|v-1|=1/R\}$ with $R$ large enough we can make the ratio arbitrarily close to
(1-p) and hence
\begin{equation}
\left|\frac{1-p (w-1)/w}{1-p (v-1)/v}\right|<1
\end{equation}
uniformly in $\Gamma_{0,1}\times\Gamma_1$. Therefore the interchange is allowed. After the summation we have
\begin{eqnarray}  \label{ortsum}
&& \sum_{\tau \in \mathbb{Z}_{\geq x}}\Phi^n_i(\tau)\Psi^n_j(\tau) \\
&&= \oint_{\Gamma_{1}}\frac{dv}{2 \pi \mathrm{i} } \oint_{\Gamma_{0,v}}\frac{%
dw} {2 \pi \mathrm{i} } \frac{\left(\frac{w-1}{w}\right)^k \left(\frac{v}{v-1%
}\right)^{j+1}} {v w (w-v)}\left(\frac{w}{v}\right)^{N}\left(\frac{w q+p}{v q+p}\right)^x,  \notag
\end{eqnarray}
where contour $\Gamma_{0,v}$ contains the pole $w=v$ inside for each $v\in \Gamma_1$, i.e. the whole $\Gamma_1$
is inside $\Gamma_{0,v}$. From $0\leq k,j< N$ we see that $w=1,0$ are not poles of the expression under the
integral. Therefore we can perform the integration in $w$ taking into account the contribution of the remaining
simple pole $w=v$, which yields
\begin{equation}
\mathrm{(}\ref{ortsum})=\oint_{\Gamma_{1}}\frac{dv}{2 \pi \mathrm{i} v^2 } \left(\frac{v-1} {v}\right)^{k-j-1}.
\end{equation}
Making the substitution $z=(v-1)/v$, we obtain the desired orthogonality (%
\ref{orthogonality})
\begin{equation}
\mathrm{(}\ref{ortsum})=\oint_{\Gamma_{0}}\frac{dz}{2 \pi \mathrm{i} } z^{k-j-1}=\delta_{k,j}.
\end{equation}
\end{proof}

Now we note that assumption (A) is fulfilled,
\begin{equation}
\Phi_0^n(\tau)=p=\phi_n(\tau_{n+1}^n,\tau),
\end{equation}
and we can write the kernel. Observe that the function $\phi_n(x,y)$ can be written in the form
\begin{equation}  \label{phi_n_int}
\phi_n(x,y)=\frac{p}{2\pi i}\oint_{\Gamma_{0,-1}}dw\frac{1}{w(w+1)^{x-y}}.
\end{equation}
Several convolutions with itself result in
\begin{eqnarray}  \label{phi^n_int}
\phi^{(n_1,n_2)}(\tau_1,\tau_2)&=& \mathbbm{1}(n_2>n_1)p^{n_2-n_1}\oint_{%
\Gamma_{0,-1}}\frac{dw}{2\pi i}  \notag \\
& \times&\left(\frac{w+1}{w}\right)^{n_2-n_1} (w+1)^{\tau_2-\tau_1-1}.
\end{eqnarray}
After a variable change
\begin{equation*}
w=\left(1-p\frac{z-1}{z}\right)^{-1}-1
\end{equation*}
we arrive at
\begin{eqnarray}  \label{phi^n_int_1}
\phi^{(n_1,n_2)}(\tau_1,\tau_2)&=&\mathbbm{1} (n_2>n_1)p \oint_{\Gamma_{1,0}}%
\frac{dz}{2\pi i z^2}  \notag \\
&\times&\left(\frac{z-1} {z}\right)^{n_1-n_2} \left(1-p\frac{z-1}{z}%
\right)^{\tau_1-\tau_2-1}.
\end{eqnarray}

Now let us perform the summation in (\ref{Prop4.3.2}) using the integral
representations (\ref{Psi_int},\ref{Phi_int}) for $\Psi_k^n(\tau)$ and $%
\Phi_k^n(\tau)$ respectively. As the integral representation of $%
\Phi_k^n(\tau)$ is identically zero for $k<0$, the summation can formally be extended to $k=\infty$. Then, we
bring the summation under the integrals, which is allowed provided that
\begin{equation}
\left|\frac{w(v-1)}{v(w-1)}\right|<1
\end{equation}
uniformly in $\Gamma_{0,1}\times\Gamma_1$. This is satisfied, for example, for $\Gamma_{0,1}=\{|w|=R\}$ and
$\Gamma_{1}=\{|v-1|=1/R\}$ for $R$ large enough. The contours can be continuously deformed without crossing the
poles. In particular, contour $\Gamma_1$ is constrained to be inside contour $\Gamma_{1,0}$.

Then, the summation yields
\begin{eqnarray}  \label{sum phi psi}
&&\sum_{k=1}^{\infty}\Psi_{n_1-k}^{n_1}(\tau_1)\Phi_{n_2-k}^{n_2}(\tau_2) \\
&&=p \oint_{\Gamma_1}\frac{dv}{2\pi i v}\oint_{\Gamma_{0,v}}\frac{dw}{2\pi i
w} \frac{(1-p(\frac{w-1}{w}))^{\tau_1}(\frac{w-1}{w})^{n_1}(w/v)^{x+N}} {%
(1-p(\frac{v-1}{v}))^{\tau_2+1}(\frac{v-1}{v})^{n_2}(w-v)} .  \notag
\end{eqnarray}
After substituting (\ref{phi^n_int_1}) and (\ref{sum phi psi}) into (\ref%
{Prop4.3.2}) we arrive at the final expression for the correlation kernel. The result can be summarized as
follows.

\begin{proposition}
The correlation kernel of the measure $\mathcal{M}$, (\ref{det_measure}), is
\begin{eqnarray}  \label{cor_kernel}
&&K(n_1,\tau_1;n_2,\tau_2)= \\
&&p \oint_{\Gamma_1}\frac{dv}{2\pi i v}\oint_{\Gamma_{0,v}}\frac{dw}{2\pi i w%
} \frac{(1-p(\frac{w-1}{w}))^{\tau_1}(\frac{w-1}{w})^{n_1}(w/v)^{x+N}} {(1-p(%
\frac{v-1}{v}))^{\tau_2+1}(\frac{v-1}{v})^{n_2}(w-v)}  \notag \\
&&-\mathbbm{1}(n_2>n_1) \oint_{\Gamma_{1,0}}\frac{p dz}{2\pi i z^2} \left(%
\frac{z-1}{z}\right)^{n_1-n_2} \left(1-p\frac{z-1}{z}\right)^{\tau_1-%
\tau_2-1}.  \notag
\end{eqnarray}
\end{proposition}

\section{TASEP current correlations}
In present section we obtain the current correlation function of the TASEP with step initial conditions and study
its asymptotical behaviour. Note that the dynamics of particles is independent of the particles on the
left. Therefore, one can limit the consideration to $N$ rightmost particles. The previous results furnish us with all the means required to write the quantity of interest (\ref{joint distrib}) as a Fredholm determinant, which is the content of the following theorem:

\begin{theorem} Given $N$ particles starting at the same moment of time $t=0$ at  sites $x_i^0=1-i, i=1,2,\dots,N$, let  $t_{n_1} \leq t_{n_2} \leq \dots \leq t_{n_m}$ be the moments of time, when $m\leq N$ particles labeled by indices $n_1<n_2<\dots<n_m$ jump off positions $x_{n_1},\dots x_{n_m}$, respectively, where $x_i=x+N-i$.
Then, their joint distribution $\mathbf{P}$, defined in (\ref{joint distrib}),
 can be represented as Fredholm determinant
\begin{eqnarray}
\mathbf{P}=\det\left(\mathbbm{1}-\chi_a K \chi_a)\right)_{l^2(\{n_1,\dots,n_m\}\times \mathbbm{Z}_{\geq x})},
\label{fredholm}
\end{eqnarray}
where $\chi_a(n_i)(t)=\mathbbm{1}(t>a_i)$.
\end{theorem}

\begin{proof}
Having written the correlation kernel for the measure $\mathcal{M}$, we can use the correspondence stated in the
Proposition \ref{GGF-measure} to write
\begin{eqnarray}
\mathbf{P}&\equiv &Prob\left(\{t_{n_1}\leq a_1\}\bigcap \{t_{n_2}\leq
a_2\}\bigcap \dots \bigcap \{t_{n_m}\leq a_m\}\right)  \notag \\
&=&\mathcal{M}\left(\mathcal{T}\supset\{\tau_1^{n_1}\leq a_1\}\bigcap \dots \bigcap \{\tau_1^{n_m}\leq
a_m\}\right),
\end{eqnarray}
By the exclusion-inclusion principle it is written as a sum
\begin{equation}  \label{fredholm1}
\mathbf{P} =\sum_{n\geq 0}\frac{(-1)^n}{n!}\sum_{i_1,\dots,i_n=1}^m
\sum_{\tau_{1}>a_{i_1}}\dots\sum_{\tau_{n}>a_{i_n}}\det\{K(n_{i_k},%
\tau_k;n_{i_j}, \tau_j)\}_{k,j=1}^n
\end{equation}
which yields the Fredholm determinant. The formula holds, provided that $%
\chi_a K \chi_a$ is a trace class operator. We do not give a proof of trace class property here and refer the
reader to similar proofs in other papers \cite{BorodinFerrari,BFS}.
\end{proof}

\subsection{Scaling behaviour}

We are interested in the behaviour of the Fredholm determinant in a scaling limit. Let us look at a bulk particle
with number
\begin{equation}  \label{n}
n=\nu L
\end{equation}
jumping off the position
\begin{equation}  \label{x_n}
x_n=\chi L
\end{equation}
at time
\begin{equation}  \label{t_n}
t_n=\omega L
\end{equation}
where $L$ is a large parameter. In the limit $L \to \infty$ the random
variable $\omega$ becomes deterministic, in a way related to the variables $%
\gamma$ and $\nu$.

According to the law of large numbers for stochastic particle dynamics \cite{Rost,Sepp}
the density of particles $\rho(x,t)$ solves the continuity equation
\begin{equation}  \label{euler}
\frac{\partial \rho(x,t)}{\partial t}+ \frac{\partial j(x,t)}{\partial x}=0,
\end{equation}
where $j(x,t)$ is a stationary state current, which for the TASEP with
backward sequential update is given by
\begin{equation}  \label{current}
j=\frac{p \rho(1-\rho)}{1-p \rho}.
\end{equation}
Solving (\ref{euler}), (\ref{current}) with the initial condition $\rho(x,0)=%
\mathbbm{1}(-x)$ we get the macroscopic density profile
\begin{equation}
\rho(x,t)=\left\{
\begin{array}{ll}
1, & x/t<p/(p-1) \\
\frac{1}{p}\left(1-\sqrt{\frac{q}{1-x/t}}\right), & p/(p-1) \leq x/t < p \\
0, & x/t\geq p.%
\end{array}
\right.
\end{equation}

At a given time step the number $n$ and the coordinate $x_n$ of a particle are related via the density
\begin{equation}
n=\int_{x_n}^\infty \rho(x,t)dx.
\end{equation}
Under the scaling (\ref{n})-(\ref{t_n}) we arrive at the relation
\begin{equation}
\omega=\frac{1}{p}\left(\sqrt{q\nu}+\sqrt{\chi+\nu}\right)^2,
\end{equation}
which holds when $p/(p-1)\leq\chi / \omega \leq p$. For a rigorous proof of this relation one can
use \cite{RS} and note that in our case the final positions of particles are fixed as $x_n=x+N-n$. Hence the
ratio
\begin{equation}
\gamma\equiv \frac{x+N}{L},
\end{equation}
is fixed as $L$ goes to infinity and therefore
\begin{equation}
\gamma=\chi+\nu.
\end{equation}
This gives the function \cite{RS}
\begin{equation}  \label{omega(nu)}
\omega(\nu)=\frac{1}{p}\left(\sqrt{q\nu}+\sqrt{\gamma}\right)^2,
\end{equation}
where $\nu\geq 0,\; \gamma \geq 0$, which describes the most probable time the $n$-th particle jump off the site
$x_n$.

Under the scaling limit of correlation functions we understand the correlations between the jump off times for
particles with the numbers close to a certain point on the deterministic scale, being apart from each other on
the scale characteristic for the KPZ class where non-trivial scaling behaviour is expected. One chooses
$|n_i-n_j|\sim L^{2/3}$. Then the dominating scale of time fluctuations is $\delta t_i \sim L^{1/3}$. This
scale defines the domain, where main contribution into the sum (\ref%
{fredholm1}) comes from.

A rigorous analysis consists of several steps. One has to prove that the kernel converges uniformly on bounded
sets to its scaling limit and that the part of the sum (\ref{fredholm1}) coming from the complement to these sets
is negligible. We do not follow this program here and limit ourselves to a simple saddle point analysis,
outlining a sketch of the proof of the kernel convergence to the $\mathrm{Airy}_2$ kernel. Rigorous mathematical
details for similar cases can be found in \cite{BFP,BorodinFerrari}.

\begin{lemma}
\label{scaling limit} Let us introduce scaling variables such that
\begin{eqnarray}
n_i&=&[L+u_i L^{2/3}] \\
\tau_i&=&[L \omega(1+u_i L^{-1/3})+L^{1/3}s_i]
\end{eqnarray}
with $u_i,s_i,\; i=1,2$ being fixed as $L\to\infty$, and let $\tilde{K} = UKU^{-1}$ with kernel $K$ given by
Proposition \ref{cor_kernel} and
\begin{equation}
U_{i,j}(s_i,s_j)=e^{-(Nf_{\nu_i}(w_i)+N^{1/3}s_i g(w_i))}\delta_{i,j}\delta(s_i-s_j).  \label{U}
\end{equation}
Then
\begin{equation}
\lim_{L\to\infty} L^{1/3} \tilde{K}(n_1,\tau_1;n_2,\tau_2) = \kappa_t K_{%
\mathrm{Airy}_2} (\kappa_h u_1,\kappa_t s_1; \kappa_h u_2,\kappa_t s_2),
\end{equation}
where in the r.h.s. we have the extended Airy kernel,
\begin{eqnarray}
&K_{\mathrm{Airy}_2} &(\xi_1,\zeta_1; \xi_2,\zeta_2) \\
&&=\left\{
\begin{array}{ll}
\int_0^\infty d\lambda e^{\lambda(\xi_2-\xi_1)}\mathrm{Ai}(\lambda+\zeta_1)%
\mathrm{Ai} (\lambda+\zeta_2), & \xi_2\leq\xi_1 \\
-\int_{-\infty}^0 d\lambda e^{\lambda(\xi_2-\xi_1)}\mathrm{Ai}%
(\lambda+\zeta_1)\mathrm{Ai} (\lambda+\zeta_2), & \xi_2 > \xi_1%
\end{array}%
\right. ,  \notag
\end{eqnarray}
and nonuniversal scaling constants
\begin{eqnarray}
\kappa_h= \frac{ q^{1/6}\gamma^{1/3}}{2(1+\sqrt{q\gamma})^{1/3} (\sqrt{\gamma%
}+\sqrt{q})^{1/3}}  \label{kappa_h} \\
\kappa_t= \frac{p  q^{-1/6}\gamma^{1/6}}{(1+\sqrt{q\gamma})^{2/3} (\sqrt{%
\gamma}+\sqrt{q})^{2/3}} .  \label{kappa_t}
\end{eqnarray}
\end{lemma}

\begin{proof}
Suppose
\begin{eqnarray}
n_i &=&[N \nu_i ] \\
\tau_i &=& [N \omega(\nu_i)+N^{1/3}s_i ] \\
x&=&[N(\gamma-1)]
\end{eqnarray}
where
\begin{equation}
\nu_i=\nu+u_i N^{-1/3},
\end{equation}
and $s_i,u_i\in\mathbbm{R}$, $i=1,2$. As the numbers of particles are constrained to be below $N$, formally the
value of $\nu$ should be in the range $0<\nu<1$. However, as it will be clear below, this constraint is
superficial, and can be omitted.

We also introduce the following functions
\begin{eqnarray}
f_{\nu}(w)&=&\omega(\nu)\ln(q+p/w)+\nu \ln(1-1/w)+\gamma \ln(w) \\
g(w)&=&\ln(q+p/w) \\
h(w)&=&\ln(1-1/w)
\end{eqnarray}
The position of the double critical point of the function $f_\nu(w)$, which satisfies $f^{\prime
}_\nu(w_0)=f^{\prime \prime }_\nu(w_0)=0$ is
\begin{equation}  \label{crit point}
w_0(\nu)=1+\sqrt{\frac{\nu}{q \gamma}}.
\end{equation}
To obtain the asymptotics of the double integral part of the kernel, we
recall that it is obtained as a sum $\sum_{k=1}^{\infty}\Psi_{n_1-k}^{n_1}(%
\tau_1)\Phi_{n_2-k}^{n_2} (\tau_2)$. Therefore we first evaluate the integrals for $\Psi_{n_1-k}^{n_1}(\tau_1)$
and $\Phi_{n_2-k}^{n_2}(\tau_2)$ asymptotically as $N\to\infty$ and then perform the summation. In terms of above
notations the integrals read
\begin{eqnarray}
\Psi^{n_1}_{n_1-k}(\tau_1))&=&\oint_{\Gamma_{0,1}}\frac{dw}{2\pi\mathrm{i}
w^2} e^{N f_{\nu_1}(w)+N^{1/3}\left(s_1 g(w)+r h(w)\right)} \\
\Phi^{n_2}_{n_2-k}(\tau_1))&=&\oint_{\Gamma_{1}}\frac{p dw }{2\pi \mathrm{i}}%
\frac{e^{-N f_{\nu_2}(w)-N^{1/3}(s_2 g(w)+r h(w))}}{ (w-1)(qw+p)}
\end{eqnarray}
where $r=k N^{-1/3}$.

Instead of the functions in the exponents we use their Taylor expansion at the double critical points $w_i\equiv
w_0(\nu_i)$, with $i=1,2$ where the main contribution to the integral comes from for large $N$. To lowest
non-vanishing order we have
\begin{eqnarray}
f_{\nu_i}(w)& = &f_{\nu_i}(w_i)+\frac{1}{6}f^{^{\prime \prime \prime
}}_{\nu}(w_0)(w-w_i)^3 + \dots \\
g(w)& = &g(w_i)+g^{\prime }(w_0)(w-w_i) + \dots \\
h(w)& = &h(w_i)+h^{\prime }(w_0)(w-w_i) + \dots
\end{eqnarray}
where in the coefficients of the $w$-dependent terms we, without loss of accuracy, replace $\nu_i$ and $w_i$ by
$\nu$ and $w_0\equiv w_0(\nu)$ respectively. We substitute these expansion into the integrals, and choose
steepest descent contours such that they approach the horizontal axis at the points $w_1$ and $w_2$ at the angles
$\pm \pi/3$ and $\pm 2 \pi/3$
respectively. Changing the integration variables to $\xi_i=(w-w_i)N^{1/3}f^{%
\prime \prime \prime }(w_0)/2$ we arrive at the integrals defining the Airy functions:
\begin{equation}
\mathrm{Ai}(x)= \int_{\infty e^{-\mathrm{i}\pi/3}}^{\infty e^{\mathrm{i}%
\pi/3}} \frac{dx}{2 \pi \mathrm{i}}\exp\left(\frac{x^3}{3}-xz\right).
\end{equation}
As a result we have the asymptotic behaviour for large $N$
\begin{eqnarray}  \label{Phi_airy}
\Psi^{n_1}_{n_1-k}(\tau_1))&\sim &\frac{e^{N f_{\nu_1} (w_1)+N^{1/3}(s_1
g(w_1)+r h(w_1))}}{w_0^2 (Nf_\nu^{^{\prime \prime \prime }}(w_0)/2)^{1/3}} \\
&\times& \mathrm{Ai}\left(\frac{rh^{\prime }(w_0)-s_1 g^{\prime }(w_0)}{%
(f_\nu^{^{\prime \prime \prime }}(w_0)/2)^{1/3}}\right)  \notag \\
\Phi^{n_2}_{n_2-k}(\tau_1))&\sim &\frac{pe^{-N f_{\nu_2}(w_2)- N^{1/3}(s_2 g(w_2)+r
h(w_2))}}{(w_0-1)(qw_0+p)(Nf_\nu^{^{\prime \prime \prime
}}(w_0)/2)^{1/3}} \\
& \times&\mathrm{Ai}\left(\frac{rh^{\prime }(w_0)-s_2 g^{\prime }(w_0)}{%
(f_\nu^{^{\prime \prime \prime }}(w_0)/2)^{1/3}}\right).  \notag
\end{eqnarray}

Finally, the summation over $k$ can be replaced by an integration over $r$. This requires arguments that the
convergence of $\Psi_{n_1-k}^{n_1}(\tau_1)$ and $\Phi_{n_2-k}^{n_2}(\tau_2)$ is uniform on the sets $k<A N^{1/3}$
and terms of the sum for higher $k$ are negligible for $A$ large enough. These arguments are based on the
superexponential decay of Airy functions. To perform the summations we use one more expansion:
\begin{equation}
h(w_i)=h(w_0)-h^{\prime }(w_0)w^{\prime }_0(\nu)u_iN^{-1/3} +O(N^{-2/3}).
\end{equation}
This yields the large $N$ behaviour of the sum as
\begin{eqnarray}
&&\sum_{k=1}^{\infty}\Psi_{n_1-k}^{n_1}(\tau_1)\Phi_{n_2-k}^{n_2}(\tau_2) =
\label{Airy sum} \\
&&N^{-1/3}\widetilde{\kappa}_t e^{({N(f_{\nu_1}(w_1)-f_{\nu_2}(w_2))+
N^{1/3}(s_1 g(w_1)-s_2 g(w_2))})}  \notag \\
&& \times \int_0^\infty d\lambda e^{\lambda \widetilde{\kappa}_h(u_2-u_1)}
\mathrm{Ai}\left(\lambda+\widetilde{\kappa}_ts_1\right) \mathrm{Ai}%
\left(\lambda+\widetilde{\kappa}_ts_2\right) ,  \notag
\end{eqnarray}
where
\begin{eqnarray}
\widetilde{\kappa}_h\!\!\!&=& \!\!\! \frac{w^{\prime }_0(\nu)
f_\nu^{^{\prime \prime \prime }}(w_0)^{1/3}} {2^{1/3}}=\frac{%
\nu^{-2/3}q^{1/6}\gamma^{1/3}}{2(\sqrt{\nu}+ \sqrt{\gamma q})^{1/3}(\sqrt{%
\gamma}+\sqrt{\nu q})^{1/3}} \\
\widetilde{\kappa}_t &=& - \frac{ 2^{1/3} g^{\prime }(w_0)}{f_\nu^{^{\prime
\prime \prime }}(w_0)^{1/3}}= \frac{p \nu^{1/6}q^{-1/6}\gamma^{1/6}}{(\sqrt{\nu%
}+\sqrt{\gamma q})^{2/3} (\sqrt{\gamma}+\sqrt{\nu q})^{2/3}}
\end{eqnarray}

Let us now evaluate the second part of the kernel given by the single integral, which can be written as
\begin{eqnarray}
I=p\oint\frac{dz}{2 \pi \mathrm{i }z^2} e^{N(f_{\nu_1}(z)-f_{\nu_2}(z))+ N^{1/3}(s_1-s_2)g(z)}
\end{eqnarray}
The critical point of the exponentiated function is found to be
\begin{equation}
z_c=w_0\equiv w_0(\nu)
\end{equation}
We use the Taylor expansion around this point to catch the main contribution to the integral. Keeping all terms
up to the relevant order we obtain
\begin{eqnarray}
f_{\nu_1}(z)-f_{\nu_2}(z) & = & f_{\nu_1}(w_1)-f_{\nu_2}(w_2) \\
&+& \frac{f^{\prime \prime \prime }(w_0)(u_2^3-u_1^3)w_0^{\prime 3}}{6 N}
\notag \\
&+&\frac{f^{\prime \prime \prime }(w_0)(u_1^2-u_2^2)w_0^{\prime 2}}{2 N^{2/3}%
}(z-w_0)  \notag \\
&+&\frac{f^{\prime \prime \prime }(w_0)(u_2-u_1) w_0^{\prime }(\nu)}{2 N^{1/3}}(z-w_0)^2 + \dots  \notag
\end{eqnarray}
and
\begin{eqnarray}
(s_1-s_2)g(z) & = & s_1 g(w_1)-s_2g(w_2) \\
&+&(u_2 s_2-u_1 s_1)N^{-1/3}g^{\prime }(w_0)w^{\prime }(\nu)  \notag \\
&+&(s_1-s_2)g^{\prime }(w_0)(z-w_0) + \dots  \notag
\end{eqnarray}
Substituting these expansions into the integral and integrating along the vertical line crossing the horizontal
axis at $w_0$ we obtain:
\begin{eqnarray}
I=N^{-1/3}\widetilde{\kappa}_t e^{{N(f_{\nu_1}(w_1)-f_{\nu_2}(w_2))+
N^{1/3}(s_1 g(w_1)-s_2 g(w_2))}} \\
\frac{e^{\frac{\widetilde{\kappa}_h^3(u_2^3-u_1^3)}{3}- \frac{(\widetilde{%
\kappa}_h^2(u_1^2-u_2^2) -\widetilde{\kappa}_t(s_1-s_2))^2}{4\widetilde{%
\kappa}_h(u_2-u_1)}-\widetilde{\kappa}_h \widetilde{\kappa}_t(s_2u_2-s_1u_1)}%
}{\sqrt{4\pi\widetilde{\kappa}_h(u_2-u_1)}}.
\end{eqnarray}

One can see that the first line of this expression exactly coincides with the factor before the integral in
(\ref{Airy sum}). Furthermore, its exponential part does not change the value of the determinants, so that it
can be omitted. The second part can be rewritten using the formula from \cite%
{johansson2}
\begin{eqnarray}
 \frac 1{\sqrt{4\pi (\tau'-\tau)}}
  e^{-(\xi-\xi')^2/4(\tau'-\tau)
  -(\tau'-\tau)(\xi+\xi')/2+(\tau'-\tau)^3/12}\\=\int_
{-\infty}^\infty e^{-\lambda(\tau-\tau')} \mathrm{Ai}(\xi+\lambda)\mathrm{Ai}(\xi'+\lambda)d\lambda,
\end{eqnarray}
where we should set $\tau=\widetilde{\kappa}_h u_1,\tau^{\prime }=\widetilde{%
\kappa}_h u_2, \xi=\widetilde{\kappa}_t s_1,\xi^{\prime }=\widetilde{\kappa}%
_t s_2$. Notice that there is a common factor
\begin{equation}
A_{1,2}=\exp(N(f_{\nu_1}(w_1)-f_{\nu_2}(w_2))+N^{1/3}(s_1g(w_1)-s_2g(w_2)) \label{A_{1,2}}
\end{equation}
in (6.38) and (6.45). Multiplication of the kernel by the inverse of that factor is equivalent to a similarity
transformation $UKU^{-1}$ of the operator $K$ with $U$ given by (\ref{U}).

The similarity transformation does not affect the determinants \ $%
\det\{K(u_i,s_i;u_j,s_j)\}_{1\leq i,j\leq n}$ and hence the resulting Fredholm determinant remains unchanged. In
addition, after this multiplication the limit of the kernel multiplied by $N^{1/3}$ is well defined. As a result,
we obtain
\begin{equation}
\lim_{t\to\infty}N^{1/3}A_{1,2}^{-1}((6.38) + (6.45))=\widetilde{\kappa}_t
K_{\mathrm{Airy}_2}(\widetilde{\kappa}_h u_1,\widetilde{\kappa}_t s_1;%
\widetilde{\kappa}_h u_2,\widetilde{\kappa}_t s_2).
\end{equation}

Finally, we set the large parameter
\begin{equation}
L=\nu N
\end{equation}
which removes the dependence of $\widetilde{\kappa}_h,\widetilde{\kappa}_t$
on $\nu$ resulting in $\kappa_h,\kappa_t$ shown in (\ref{kappa_h}), (\ref%
{kappa_t}). The same result could be obtained by setting $\nu=1$ in the very beginning. The independence of the
result on $\nu$ follows from the fact that the dynamics of a particle is independent on particles behind it,
which is specific of the TASEP.
\end{proof}

The constant $\kappa_t$ has occurred before in \cite{Johansson,RS}. The Lemma together with the steps mentioned
above would result in the following theorem describing the universal behaviour of the rescaled process

\begin{theorem}\label{Airy_2}
The following limit holds in a sense of finite-di\-men\-sion\-al distributions:
\begin{equation}
\lim_{L \to \infty} \frac{t_{L+u L^{2/3} }-L\omega(1 +u L^{-1/3})}{%
L^{1/3}}=\kappa_t \mathcal{A}_2(\kappa_h u),
\end{equation}
where $\mathcal{A}_2$ is the $Airy_2$ process characterized by multipoint distributions:
\begin{eqnarray}
\mathrm{Prob}(\mathcal{A}_2(u_1)<s_1,\dots,\mathcal{A}_2(u_m)<s_m)  \notag \\
=\det\left(\mathbbm{1}-\chi_{ s} K_{\mathrm{Airy_2}} \chi_{s})\right)_{L^2(\{n_1,\dots,n_m\}\times \mathbbm{R})}.
\end{eqnarray}
\end{theorem}

\section{Conclusion and perspectives}

In the present work we studied multi-point correlations in the discrete time TASEP with backward sequential
update. Using the GGF, we have extended the range of analysis
to  space-time configurations, which were not covered by the previous studies.
With the help of the GGF  and the notion of boundaries one can describe
the joint probabilities of both the arrivals
at given sites  by given times for some particles and the hoppings from  given sites at given times for the others.
The arrangement of points in admissible configurations, dealt with by the GGF,
is defined by the strictly decreasing  order of space coordinates, (\ref{x order}),
 and weakly increasing order of time coordinates, (\ref{t order}).
 The latter is exactly opposite
 to the definition of space-like ordering accepted in    \cite{BorodinOlshanski,BFS,BorodinFerrari}. As such,
 it extends the set of space-time configurations, where the probability measure having a compact
 determinantal form can be used as a starting point for calculation of correlation functions.
 However, the  constraint (\ref{x order}) is an additional condition, which makes the set of
 point configurations living on the boundaries narrower than the whole set of time-like configurations
 complementary to the set of space-like  ones.

 We considered the simplest example of the use of the GGF for the
 calculation of correlation functions. This is the case when the boundary
 consists of consecutive vertical lines in the space-time plane
 and the quantity calculated is the joint probability for selected particles
 to jump from corresponding lines by given time moments.
The expression for the kernel (\ref{cor_kernel}) we derived in Section 5 should be compared with that obtained in
\cite{ImamSasamoto} for the dynamics of a tagged particle in the TASEP with the step initial conditions. Indeed,
in the case of strictly ordered times $t_1<t_2<\dots<t_m$ , all endpoints in the joint distribution correspond to
reference points of a hole moving in the opposite direction.
For the TASEP with parallel update where the
particle-hole symmetry holds, the hole can be considered as the tagged particle.
Therefore the difference in the kernels must be attributed to the difference in the updates,
parallel and backward sequential.
One also notices a partial similarity  between the auxiliary
determinantal weights obtained in our case, (\ref{det_measure}), and proposed in \cite{ImamSasamoto} for the case of
current correlations in the TASEP with parallel update (see formula (6.7) of \cite{ImamSasamoto} ).
They have similar structure of products of Schur functions, except one term which seems to have more complicated
form in our case. It was also communicated to us by A. Borodin that the auxiliary determinantal
measure (\ref{det_measure}) is a
particular case of multi-parametric family of measures studied in \cite{BorodinPeche}.
One of the questions to answer is how these parameters can be incorporated into the particle dynamics.

In full generality, the method based on the GGF can be applied to obtain
the correlation functions on more general sets.
There are several ways for generalization. One can obtain an auxiliary determinantal process
associated to an arbitrary boundary $\mathcal{B}$.
Then, one can try to construct a cascade-like determinantal process similar to that in
\cite{BorodinFerrari}  using a series of subsequent boundaries. At simplest, these could be
the fixed coordinate boundaries studied here:
\begin{equation}  \label{bound_k}
\bm{\mathcal{B}}^{(j)}=\bigcup_{i=1}^N\{(x^{(j)}-i+N,t_i):t_i\in\mathbb{Z}\}_i%
\end{equation}
with $x^{(1)}<x^{(2)}<\dots$, the upper index $j$ counting the number of points within the correlation functions.
The first method would work with general admissible configurations of different particles,
while the second would allow to relax time ordering constraint (\ref{t order}) and
to consider several space-time points for a single particle.
Another generalization of interest is an analysis of the TASEP with different initial conditions.
Note that the  GGF allows varying not only initial space positions of particles but also the
 time moments when the particles enter the dynamics.

The asymptotic analysis of the kernel (\ref{cor_kernel}) we obtained
shows that current correlation functions converge to that of the ${Airy}_2$ process.
This result agrees with the slow decorrelation
 \cite{Ferrari} observed in the models of the KPZ class in
$1+1$ dimensions. For the same reason we expect that a similar result also holds for general
boundaries and for the configurations accessible
for multi-cascade analysis. The most intriguing problem is the correlation along the characteristic lines, where
they are expected to be different from standard KPZ limiting processes.
A primary analysis shows that the characteristic lines stay beyond the domain satisfying the $x$-ordering constraint
(\ref{x order}). Therefore the technique used here has to be sufficiently modified to advance in this direction.

\section{Acknowledgements}
The authors are grateful to Tony Dorlas for his help in writing Section 5. The authors thank
Patrik Ferrari and Alexei Borodin for stimulating discussions.
This work was supported by the RFBR grants 07-02-91561a, 09-01-00271a and
the DFG grant 436 RUS 113/909/0-1(R). AMP thanks Dublin Institute for Advanced Studies for hospitality.

\appendix

\section{Cyclic structure of the de\-ter\-mi\-nants $G_0(x,y)$ and $G_1(x,y)$\label%
{cycles appendix}}

\begin{lemma}
\label{cycles} Let $\sigma=(\sigma_1,\dots,\sigma_N)\in S_N$ be a permutation, $\bm{x}$ and $\bm{y}$ be two
particle configurations from the range (\ref{range}) and
\begin{equation}
\mathcal{F}_\sigma=\prod_{i=1}^N F_{\sigma_i-i}(x_i-y_{\sigma_i},t) . \label{F_sigma}
\end{equation}

\begin{enumerate}
\item Then, if $t=0$, $\mathcal{F}_\sigma=0$ unless $\sigma$ is the identity
permutation,
\begin{equation}
\sigma=(1,\dots,N)
\end{equation}
and
\begin{equation}
\bm{x}=\bm{y}.
\end{equation}

\item Let $t=1$, and $(i_1,\dots,i_k)$ be a cycle consisting of $k>1$
elements of $\sigma$, such that $\sigma_{i_{1}}=i_{2},\sigma _{i_{2}}=i_{3},\dots ,\sigma _{i_{k}}=i_{1}$. Then,
$\mathcal{F}_\sigma=0$ unless up to a cyclic shift of indices
\begin{equation}
i_2=i_1+1,\dots,i_k=i_1+k-1
\end{equation}
and
\begin{equation}
y_{i_1}=x_{i_1},x_{i_2}=y_{i_2}=x_{i_1}-1,\dots,x_{i_k}=y_{i_k}=x_{i_1}-k+1.
\end{equation}
\end{enumerate}
\end{lemma}

\begin{proof}
We first note that
\begin{equation}
F_n(x,t)=0  \label{F_n=0}
\end{equation}
unless
\begin{equation}
n \leq x \leq t \,\,\,\mathrm{for}\,\,\,n\leq 0  \label{n<0}
\end{equation}
and
\begin{equation}
x \leq t \,\,\,\mathrm{for}\,\,\, n>0 .  \label{n>0}
\end{equation}

For $\mathcal{F}_\sigma$ to be nonzero all the factors in the r.h.s. of (\ref%
{F_sigma}) have to be nonzero simultaneously. Then, from (\ref{n<0}), (\ref%
{n>0}) the following inequalities hold:
\begin{eqnarray}
\sigma_{i}-i \leq x_i-y_{\sigma_i}\leq t,\qquad &\mathrm{for\quad }& \sigma
_{i}-i\leq 0  \label{i>sigma_ineq} \\
x_i-y_{\sigma_i} \leq t,\qquad &\mathrm{for\quad }& \sigma _{i}-i>0. \label{i<sigma_ineq}
\end{eqnarray}%
Moreover, the particle coordinates vary within the domain (\ref{range}), which implies
\begin{eqnarray}
i-\sigma _{i} &\leq &x_{\sigma _{i}}-x_{i},\,i-\sigma _{i}\leq y_{\sigma_{i}}-y_{i},\, \mathrm{for\quad }
\sigma_{i}-i\leq 0,
\label{i>sigma_range} \\
\sigma _{i}-i &\leq &x_{i}-x_{\sigma _{i}},\,\sigma _{i}-i\leq y_{i}-y_{\sigma _{i}},\, \mathrm{for\quad }\sigma
_{i}-i>0. \label{i<sigma_range}
\end{eqnarray}%
Comparing (\ref{i>sigma_ineq}) and (\ref{i>sigma_range}) we obtain for $%
(i\geq\sigma _{i})$
\begin{equation}
y_i\leq x_i,\,\,y_{\sigma _{i}}\leq x_{\sigma _{i}},  \label{y_sigma<x_sigma}
\end{equation}%
while from (\ref{i<sigma_ineq}) and (\ref{i<sigma_range}) for the case $%
(i<\sigma_i)$ we have
\begin{equation}
x_i-y_i<t,\,\,x_{\sigma _{i}}-y_{\sigma _{i}}<t.  \label{y_sigma>x_sigma}
\end{equation}

Consider a cycle consisting of $k>1$ elements, $(i_{1},\dots,i_{k})$.  Below we adopt the convention
$i_{k+1}\equiv i_1$ for $j=1,\dots,N$. Let us choose $r$, $1\leq r\leq k$, such that the number $i_{r}$ is the
smallest of the
numbers $i_{1},\dots ,i_{k}$. It follows that%
\begin{eqnarray}
\sigma _{i_{r-1}} &\equiv&i_{r}<i_{r-1},  \label{i_k>i_k-1} \\
i_{r} &<&i_{r+1}\equiv\sigma _{i_r}.  \label{i_k>i_k+1}
\end{eqnarray}%
Using (\ref{i>sigma_ineq}) and (\ref{y_sigma<x_sigma}) the value of $x_{i_r}$ can be bounded from above and below
respectively:
\begin{equation}
y_{i_r}\leq x_{i_{r}}\leq y_{i_{r+1}}+t.  \label{x bounds}
\end{equation}

\begin{enumerate}
\item Let $t=0$. The inequality (\ref{x bounds}) then yields
\begin{equation*}
y_{i_r}\leq y_{i_{r+1}},
\end{equation*}
which cannot be satisfied together with (\ref{i_k>i_k+1}) within the domain (%
\ref{range}). Thus we conclude that for $t=0$ only trivial cycles with $k=1$ are possible for $\mathcal{F}_
\sigma$ to be nonzero. The permutation
containing only trivial cycles is the identical permutation, $%
\sigma=(1,\dots,N)$. From
\begin{equation}
F_0(x,0)=\delta_{x,0}
\end{equation}
we obtain the first statement of the Lemma.

\item Let $t=1$. First we note that the only way to satisfy (\ref{x bounds})
is to set
\begin{equation}
y_{i_r}=x_{i_r}=y_{i_{r+1}}+1.  \label{y=x=y+1}
\end{equation}
This means that the particles with indices $i_r$ and $i_{r+1}$ are next to each other and are rightmost in the
cycle (and in the cluster), i.e.
\begin{equation}
i_r=i_{r+1}-1.  \label{rightmost}
\end{equation}
Then, applying (\ref{i>sigma_ineq}), (\ref{i>sigma_range} and (\ref{y=x=y+1}%
) to the case $i=i_{r-1}, \sigma_i=i_r$, we obtain
\begin{equation}
x_{i_r}-x_{i_{r-1}}=i_{r-1}-i_r,
\end{equation}
which suggests that all sites between $x_{i_r}$ and $x_{i_{r-1}}$ are occupied, i.e. the particles with indices
$i_r$ and $i_{r-1}$ belong to the same packed cluster of particles of the configuration $\bm{x}$.

As the cluster in the configuration $\bm{x}$ spreads to the left of $x_r$ for at least two sites, we have
\begin{equation}
x_{i_{r+1}}=y_{i_{r+1}}=x_{i_{r}}-1.
\end{equation}
If $i_{r-1}=i_{r+1}$, then $k=2$ and the cycle is completed. If $k>2$ we have
\begin{equation}
i_{r+2} \equiv \sigma_{i_{r+1}}>i_r+1.
\end{equation}
since the particles with indices $i_r$ and $i_{r+1}$ are rightmost in the cluster. We then apply
(\ref{i<sigma_ineq}) and (\ref{y_sigma>x_sigma}), which yields
\begin{equation}
x_{r+2}=x_r-2=y_{r+2}=y_r-2.
\end{equation}
This procedure can be iterated again and again until at $k$-th step we obtain
\begin{equation}
i_{r+k-1}=i_{r-1}.
\end{equation}
The cycle is completed.
\end{enumerate}
\end{proof}


\begin{thebibliography}{99}
\bibitem{BorodinFerrari} A. Borodin, P.L. Ferrari: Large time asymptotics of growth models on space-like paths I: PushASEP.
 Electronic Journal of
Probability, \textbf{13} 1380-1418 (2008).

\bibitem{BFP} A. Borodin, P.L. Ferrari,
 M. Pr{\"a}hofer: Fluctuations in the Discrete TASEP with Periodic Initial Configurations and the Airy1 Process., Int. Math. Res.
Papers \textbf{2007} rpm002(2007).

\bibitem{BFPS} A. Borodin, P.L. Ferrari, M. Pr{\"a}hofer, T. Sasamoto: Fluctuation properties of the TASEP with periodic initial configuration.
 J. Stat. Phys. \textbf{129} 1055 (2007).

\bibitem{BFS} A. Borodin, P.L. Ferrari, T. Sasamoto:
Large time asymptotics of growth models on space-like paths II: PNG and parallel TASEP.
 Comm. Math. Phys.
\textbf{283} 417 (2008).


\bibitem{BFS1} A. Borodin, P.L. Ferrari, T. Sasamoto: Transition between Airy1 and Airy2 processes and TASEP fluctuations
 Comm. Pure Appl. Math. \textbf{61}  1603-1629 (2008).

\bibitem{BorodinOlshanski} A. Borodin, G. Olshanski:Stochastic dynamics related to Plancherel measure on partitions.
AMS Transl.: Representation Theory, Dynamical Systems, and Asymptotic Combinatorics (V. Kaimanovich
and A. Lodkin, eds.)  922. MR2276098 (2006)

\bibitem{BorodinPeche} A. Borodin, S. P\'ech\'e:
Airy Kernel with Two Sets of Parameters in Directed Percolation and Random Matrix Theory.  J. Stat. Phys.
\textbf{132} 275-290 (2008)


\bibitem{BorodinRains} A. Borodin, E.M. Rains: EynardMehta theorem, Schur process, and their Pfaffian analogs
 J. Stat. Phys. \textbf{121}
291 (2005).



\bibitem{Brankov} J.G. Brankov, V.B. Priezzhev, R.V. Shelest: Generalized determinant solution of the discrete-time totally asymmetric
exclusion process and zero-range process.
 Phys. Rev. E
\textbf{69} 066136 (2004).


\bibitem{CorwinFerrariPeche} I. Corwin, P.L. Ferrari, S. P\'ech\'e: Universality of slow decorrelation in KPZ growth.
arXiv:1001.5345 (2010)

\bibitem{Derrida} B. Derrida: An exactly soluble non-equilibrium system: the asymmetric simple exclusion process.
 Phys. Rep. \textbf{301} 65 (1998).

\bibitem{Ferrari} P.L. Ferrari: Slow decorrelations in KardarParisiZhang growth. J. Stat. Mech. P07022 (2008).


\bibitem{ImamSasamoto} T. Imamura, T. Sasamoto: Dynamics of a tagged particle in the asymmetric exclusion process with the step initial condition.
 J. Stat. Phys. \textbf{128}
799 (2007)

\bibitem{Johansson} K. Johansson: Shape fluctuations and random matrices.
 Comm. Math. Phys. \textbf{209} 437 (2000)


\bibitem{johansson2} K. Johansson: Discrete polynuclear growth and determinantal processes.
 Comm. Math. Phys., \textbf{242} 277-329 (2003)
(2003)

\bibitem{KPZ} M. Kardar, G. Parisi, Y.-C. Zhang: Dynamic Scaling of Growing Interfaces.
 Phys. Rev. Lett. \textbf{56} 889892 (1986)

\bibitem{Liggett} T.M. Liggett, \textit{Stochastic Interacting Systems:
Contact, Voter and Exclusion Processes} (Springer, Berlin, 1999).

\bibitem{Macdonald} I.G. Macdonald, Symmetric Functions and Hall
Polynomials, Oxford Mathematical Monographs, (Oxford University Press, Oxford, 1995)

\bibitem{Nagao} T. Nagao, T. Sasamoto, Nucl. Phys. B \textbf{699} 487
(2004).


\bibitem{priezPRL} V. B. Priezzhev: Exact nonstationary probabilities in the asymmetric exclusion process on a ring.
 Phys. Rev. Lett. \textbf{91} 050601
(2003).
\bibitem{priezzhev_traject} V.B. Priezzhev: Non-stationary probabilities for the asymmetric exclusion process on a ring.
Pramana-J.Phys. \textbf{64} 915-925 (2005).

\bibitem{PovPri1} A.M. Povolotsky, V.B. Priezzhev: Determinant solution for the totally asymmetric exclusion process with parallel update.
 J. Stat. Mech. P07002 (2006).

\bibitem{PovPri2} A.M. Povolotsky, V.B. Priezzhev: Determinant solution for the totally asymmetric exclusion process with parallel update: II. Ring geometry.
 J. Stat. Mech. P08018
(2007).

\bibitem{update} N. Rajewsky, L. Santen, A. Schadschneider and M.
Schreckenberg: The asymmetric exclusion process: Comparison of update procedures. J. Stat. Phys. \textbf{92} No 1-2 , 151 - 194 (1998),

\bibitem{RS} A. R\'akos, G.M. Sch{\"u}tz: Current distribution and random matrix ensembles for an integrable asymmetric fragmentation process.
 J. Stat. Phys. \textbf{118} 511
(2005).


 \bibitem{Rost} H. Rost:
  Non-equilibrium behaviour of a many particle process: Density profile and local equilibria.
Probability Theory and Related Fields
\textbf{58} 41-53 (1981)

\bibitem{Sasamoto} T.Sasamoto: Spatial correlations of the 1D KPZ surface on a flat substrate.
 J. Phys. A \textbf{38} L549 (2005).

\bibitem{Sepp} T. Sepp{\"a}l{\"a}inen: A scaling limit for queues in series. Ann. Appl. Prob. \textbf{7} 855
(1997).

\bibitem{Sch1} G.M. Sch{\"u}tz: Exact solution of the master equation for the asymmetric exclusion process J. Stat. Phys. \textbf{88}, 427 (1997).


\bibitem{SchutzBook} G.M. Sch\"{u}tz, in \textit{Phase Transitions and
Critical Phenomena} Vol. 19, edited by C. Domb and J.L. Lebowitz (Academic Press, London, 2001).



\bibitem{Spohn} H. Spohn, \textit{Large scale dynamics of interacting
particles} (Springer, Berlin, 1991).

































\end{thebibliography}
\end{document}